\documentclass[12pt]{article}
\usepackage[T1]{fontenc}
\usepackage[dvips]{graphicx}
\graphicspath{{images/}}
\usepackage{verbatim}
\usepackage{amsmath,amsthm}
\usepackage{xspace}
\usepackage{pifont}
\usepackage{graphicx}
\usepackage{amssymb}
\usepackage{epic, eepic}
\usepackage{dsfont}
\usepackage{amssymb}
\usepackage{makeidx}
\usepackage{mathrsfs}
\usepackage{exscale}
\usepackage{color} 
\usepackage{overpic} 
\usepackage{bm}
\usepackage{bbm}
\usepackage{booktabs} 
\usepackage{color, colortbl}
\usepackage{subcaption}
\usepackage[round]{natbib}

\RequirePackage[colorlinks,citecolor=blue,urlcolor=blue]{hyperref}

\definecolor{Gray}{gray}{0.9}

\usepackage{amsmath,afterpage}
\usepackage{epsf}
\usepackage{graphics,color}

\def\0{\mathbf{0}}

\def\rr{\rightarrow}
\def\dr{\downarrow}

\def \< {\langle}
\def \> {\rangle}

\def\ol{\overline}

\def\beqa{\begin{eqnarray}}
\def\eeqa{\end{eqnarray}}
\def\beqas{\begin{eqnarray*}}
\def\eeqas{\end{eqnarray*}}

\newtheorem{theorem}{Theorem}[section]
\newtheorem{lemma}[theorem]{Lemma}

\newtheorem{proposition}[theorem]{Proposition}

\newtheorem{remark}[theorem]{Remark}

\newtheorem{definition}[theorem]{Definition}

\numberwithin{equation}{section}
\newcommand{\hatd}[1]{{}}

\newcommand{\bd}{\begin{displaymath}}
\newcommand{\ed}{\end{displaymath}}
\newcommand{\be}{\begin{equation}}
\newcommand{\ee}{\end{equation}}
\newcommand{\bq}{\begin{eqnarray}}
\newcommand{\eq}{\end{eqnarray}}
\newcommand{\bn}{\begin{eqnarray*}}
\newcommand{\en}{\end{eqnarray*}}

\def\P{\mathbb{P}}

\usepackage{authblk}

\title{Closed-Loop Nash Competition for Liquidity\footnote{The authors are grateful to Dan Lacker and Kasper Larsen for fruitful discussions and pointing out a number of pertinent references. The very detailed and constructive comments of two anonymous referees are also gratefully acknowledged.}}
\author[]{Alessandro Micheli\thanks{Supported by the EPSRC Centre for Doctoral Training in Mathematics of Random Systems: Analysis, Modelling and Simulation (EP/S023925/1).}}
\author[]{Johannes Muhle-Karbe}
\author[]{Eyal Neuman }

\affil[]{Department of Mathematics, Imperial College London}
 
\begin{document}

\maketitle

 \vspace{-1cm}

\begin{abstract}
We study a multi-player stochastic differential game, where agents interact through their joint price impact on an asset that they trade to exploit a common trading signal. In this context, we prove that a \emph{closed-loop} Nash equilibrium exists if the price impact parameter is small enough. Compared to the corresponding open-loop Nash equilibrium, both the agents' optimal trading rates and their performance move towards the central-planner solution, in that excessive trading due to lack of coordination is reduced. However, the size of this effect is modest for plausible parameter values.
 \end{abstract} 


\begin{description}
\item[Mathematics Subject Classification (2010):] 49N90, 91A25, 91G10, 93E20
\item[JEL Classification:] C73, C02, C61, G11, G12
\item[Keywords:]   price impact, stochastic games, closed-loop Nash equilibrium
\end{description}

\section{Introduction}

Most quantitative trading strategies are based on some form of ``signals'' about future price changes. A major obstacle to make such strategies profitable in practice is the adverse price impact caused by the rapid execution of large orders. In particular for large portfolios, this requires to strike a delicate balance between exploiting trading opportunities, but only doing so if the respective signals are strong and persistent enough to outweigh the associated trading costs.

Accordingly, a large and rapidly growing literature studies how to optimally exploit trading signals with various dynamics in the presence of superlinear trading costs as in the models pioneered by~\cite{OPTEXECAC00}, cf., e.g.,~\cite*{garleanu.pedersen.13,GARLEANU16,Car-Jiam-2016,moreau.al.17, lehalle.neuman.19,collin.al.20} and the references therein.\footnote{Linear costs corresponding to bid-ask spreads, which are the dominant frictions for somewhat smaller portfolios, are studied in~\cite{delataillade.al.12,martin.14}, for example.}  In this context, many investors use variants of the same trading signals, such as moving averages of past price changes~\citep{garleanu.pedersen.13}, order-book imbalances~\citep*{Cont:2013aa, cont14, citeulike:12820703,Car-Jiam-2016,lehalle.neuman.19}, or price-dividend ratios~\citep{barberis.00}.  Such ``crowded'' signals naturally lead to stochastic games, where the agents interact through the rates at which they draw on the same pool of liquidity. 

``Market-impact games'' of this kind have been analyzed in depth for the optimal execution of a single exogenously given order, exploiting that the corresponding optimal trading patterns are deterministic in many cases as, cf.~\cite*{brunnermeier.pedersen.05,carlin.al.07,schied.schoeneborn.09} as well as many more recent studies.\footnote{In particular, \cite*{chen.al.21a} show how to endogenize price dynamics, interest rates, and permanent price impact in a setting with deterministic trading schedules.} More recently, a number of papers also analyze the competition for liquidity between agents that trade to exploit a common trading signal~\citep*{voss.19,DrapeauLuoSchiedXiong:19,EvangelistaThamsten:20,N-V-2021}. To obtain tractable results, these papers focus on \emph{open-loop} Nash equilibria. This means that each player considers the others' actions to be fixed, when deciding whether to unilaterally deviate from the putative equilibrium. Each player's optimality condition can in turn be derived in analogy to the single-agent version of the model, with the other agents' actions acting as additional exogenous inputs. A Nash equilibrium can in turn be derived in a second step by imposing the consistency condition that all single-agent optimality conditions must hold simultaneously. For optimal trading problems with mean-reverting signals and linear price impact, this leads to multidimensional but linear systems of forward-backward stochastic differential equations, which admit closed-form solutions.

Even more tractable results obtain in the mean-field limit of many small agents, which often reduces the analysis to single-agent stochastic control problems with the average of all agents actions acting as an additional exogenous input. In the present context, such models have been studied by \cite*{CardaliaguetLehalle:18, CasgrainJaimungal:18, CasgrainJaimungal:20, FuGraeweHorstPopier:20,N-V-2021}, for example.

In contrast, much less is known about \emph{closed-loop} Nash equilibria, where other agents react to unilateral deviations from a putative equilibrium. More specifically, each agent then assumes that the feedback form of the others' controls are fixed but takes into account unilateral deviations through their impact on the state variables of the model. The controls of the other agents in turn can no longer be treated as an exogenous input for the single-agent optimization, leading to a much more involved coupling between the agents' individual optimality conditions.

In this paper, we study a stylized model for such closed-loop Nash equilibria, and analyze how this more sophisticated form of competition changes optimal trading patterns and welfare relative to the open-loop and to a central-planner solution. In order to obtain tractable results, we focus on a symmetric infinite-horizon model,\footnote{Finite-horizon models lead to analytically intractable systems of coupled differential equations, compare~\cite{carmona.yang.08}. In a different price impact model (where agents' holdings affect future expected returns), existence results for a finite horizon model have been obtained by~\cite*{chen.al.21}.} where a finite number of identical agents trade a risky asset, whose returns are partially predictable through a trading signal with Ornstein-Uhlenbeck dynamics. All agents' trading problems are intertwined through their common linear (instantaneous) price impact on the asset's execution price.

We show that the corresponding individual optimality and consistency conditions can be expressed in terms of a system of nonlinear algebraic equations. Unlike in more complex settings with endogenous price dynamics~\citep{sannikov.16,obizhaeva.wang.19}, we then prove that this system has a unique solution when the price impact parameter $\lambda$ is small enough. With this solution in hand, we can in turn verify that the proposed candidate indeed is a closed-loop Nash equilibrium, which is unique in the linear class that contains the (globally unique) open-loop and central-planner solutions. This analysis is complicated by the fact that even in the limit, the system does not admit a tractable  solution as in single-agent models~\citep{GARLEANU16}. Therefore, in order to show that our candidate value function is well defined and establish a verification theorem, we need to establish and exploit various implicit properties of the solution.

At the leading order for small price impact, the equilibrium trading rates admit asymptotic expansions, which disentangle the effects of trading costs, inventory costs and Nash competition between the agents. In our model, limited liquidity creates a negative externality, where agents only internalize the adverse effects that their price impact has on their own execution rates but not others'. Accordingly, Nash competition leads to excessive trading relative to the central-planner benchmark. Closed-loop equilibria, where agents react to out-of-equilibrium deviations, move part of the way from the open-loop solution to its central-planner counterpart. By somewhat reducing excessive trading, they in turn also reduce the ``price of anarchy'' by which the agents' optimal performance is reduced due to the lack of perfect coordination.\footnote{Unlike in~\cite{chen.al.21a}, where the difference between Nash equilibria and central planner solutions disappears over time, this difference is persistent in our model due to the changes in the investment opportunity set.}

However, both asymptotically and by the numerical computation of the exact equilibrium for realistic model parameters from~\cite{collin.al.20}, we find that the closed-loop equilibrium quantitatively lies much closer to its open-loop counterpart than to the central-planner solution. At least for the symmetric model that we consider here, this suggests that open-loop models can indeed serve as an accurate but much more tractable proxy for closed-loop Nash equilibria.

To the best of our knowledge, the present paper is the first existing example in the literature where open-loop and closed-loop equilibria are compared qualitatively and quantitatively in a finite-player game with interaction through the controls. Such comparisons were derived by \cite{CarmonaFouque} for systemic risk games, for example, which involve interactions through the state processes, but not through the controls. \cite{Lacker-Zariphopoulou, Lacker:2020aa,chen.al.21a} computed explicit closed-loop Nash equilibria for games with interactions through the controls, but the controls happen to be deterministic so that the open-loop and closed-loop equilibria coincide.

Some related work has appeared recently on convergence of the finite-player Nash equilibrium to the corresponding mean-field equilibrium in various settings.  \cite{Laur20} proved such convergence results for open-loop equilibria of games with idiosyncratic noise for each of the players. \cite{N-V-2021} studied the corresponding problem for execution games with common noise. \cite{lacker.leflem.21} and \cite{Djete21} proved the convergence of closed-loop solutions (under the a-priori assumption that they exist) to the mean-field solution in the case where each player is influenced by idiosyncratic and  common noise. In both these papers the idiosyncratic noise is crucial to establish the convergence. Even though the game in the present paper only has common noise, numerical evidence suggests that open- and closed-loop nevertheless converge to the many-player limit. A rigorous proof of this result -- including the correct formulation of the limiting model -- is a challenging open problem for future research.
 
The remainder of this article is organized as follows. In Section \ref{sec:model_setup} we present the multiplayer game. Section \ref{sec-results} contains our main results on the corresponding closed-loop Nash equilibrium. A comparison between the closed-loop, open-loop and central-planner solution is subsquently presented in Section \ref{sec-comp}. Section \ref{s:heuristics} outlines the heuristic derivation of the closed-loop equilibrium. The rigorous proofs of these results are in turn developed in Section \ref{sec-pf-cl}; the most onerous calculations are delegated to the appendix for better readability.

  \section{Model} \label{sec:model_setup}

\subsection{Financial Market}
Throughout the paper, we fix a filtered probability space $(\Omega, \mathcal F,(\mathcal F_t)_{t \geq 0}, \P)$ supporting two standard Brownian motions $(W^P_t)_{t \geq 0}$ and $(W_t)_{t \geq 0}$. We consider a financial market with two assets. The first one is safe, with price normalised to one. The other asset is risky; its  exogenous unaffected price process $(P_{t})_{t \geq 0}$ has dynamics  
\begin{equation}
dP_{t} = \mu_t dt + \sigma_P dW^P_{t},\quad P_{0} = p, 
\end{equation}
Here, the initial price level is a constant $p$ and the volatility is a positive constant $\sigma_P$; the expected returns have mean-reverting Ornstein-Uhlenbeck dynamics
\begin{equation}
\label{eq-ornstein-uhlenbeck}
d\mu_{t} = -\beta \mu_{t}dt + \sigma dW_{t},\quad \mu_{0} = m,
\end{equation}
for positive constants $\beta$, $\sigma$ and $m\in\mathbb{R}$. (We set the mean-reversion level to zero to simplify the already involved algebra below.) The current expected return $\mu_t$ can be interpreted as a signal about future price changes, such as dividend yields~\citep{barberis.00}, moving averages of past returns~\citep{garleanu.pedersen.13,martin.14}, or order-book imbalances~\citep{cont14, citeulike:12820703,Car-Jiam-2016,lehalle.neuman.19}, for example. 

\subsection{Agents}

 The assets are traded by $N \geq 2$ identical agents indexed by $n=1,\ldots,N$. Starting from an initial position $x\in\mathbb{R}$, they only adjust their risky holdings $\varphi^n=(\varphi^n_t)_{t \geq 0}$ at absolutely continuous trading rates $\dot{\varphi}^n_t=d\varphi^n_t/dt$, because their \emph{aggregate} trading activity $\dot\varphi=(\dot\varphi^1,\ldots,\dot\varphi^N)$ has an adverse linear temporary impact on the execution price:
 \be \label{def:S}
P^{\dot\varphi}_{t} = P_{t}  + \lambda \sum_{i=1}^{N} \dot\varphi^{i}_t, \qquad t \geq 0, 
\ee
for a positive constant $\lambda$. For each agent $n$, the others' trading rates are denoted by
\begin{equation*}
\dot\varphi^{-n} = (\dot\varphi^{1},\ldots, \dot\varphi^{n-1}, \dot\varphi^{n+1}, \ldots ,\dot\varphi^{N}).
\end{equation*}
As in \cite{GARLEANU16}, the agents choose their trading rates $\dot\varphi^n$ to maximize expected discounted returns over an infinite horizon, penalized for risk and transaction costs (relative to the unaffected execution prices):
\begin{equation} \label{def:FPGobjective}
J^{n}(\dot\varphi^{n}; \dot\varphi^{-n}) = \mathbb{E}\left[\int_0^\infty e^{-\rho t}\left(\mu_t \varphi^n_t -\frac{\gamma}{2}\left(\varphi^n_t\right)^2 -\lambda \dot\varphi_t^{n}\left(\sum_{i=1}^N\dot\varphi^i_t \right) \right)dt\right].
\end{equation}
Here, $\rho>0$ is the time-discount rate and $\gamma>0$ penalizes risk through a running cost on the agents' squared inventories.\footnote{As in many related studies, this is a reduced-form proxy for an exponential utility function with constant absolute risk aversion. Formally carrying out the analysis below for such a goal functional is still feasible by adapting single-agent results~\citep{muhlekarbe.al.21}. However, in the present game-theoretic context this would lead to a very high-dimensional system of nonlinear equations for which existence is unclear as in \cite{obizhaeva.wang.19}.} To ensure that all terms in these goal functionals are well defined, we focus on \emph{admissible} trading rates for which $\dot\varphi=(\dot\varphi^1,\ldots,\dot\varphi^N) \in \mathcal{A}_\rho^N$ and $\varphi=(\varphi^1,\ldots,\varphi^N) \in \mathcal{A}_\rho^N$, where $\mathcal{A}_{\rho}$ denotes the progressively measurable processes $(X_t)_{t \geq 0}$ that satisfy $\mathbb{E}\left[ \int_0^\infty e^{-\rho t} X_t^2 dt \right] <\infty$ and $\mathcal{A}_\rho^N$ is the space of $N$-dimensional vectors of such processes.

\subsection{Nash Equilibrium}

Due to their joint impact on the execution price~\eqref{def:S} of the risky asset, the agents' optimization problems~\eqref{def:FPGobjective} are intertwined. Their interaction accordingly has to be studied in a game-theoretic fashion:

\begin{definition}
Admissible trading rates $\dot{\varphi}=(\dot\varphi^1,\ldots,\dot\varphi^N)$ form a \emph{Nash equilibrium} if, given the trading rates $\dot\varphi^{-n}$ of the other agents, no agent $n=1,\ldots,N$ has an incentive to deviate from the equilibrium because $\dot{\varphi}^n$ already maximizes the corresponding goal functional~\eqref{def:FPGobjective}.
\end{definition} 

In dynamic models, different notions of Nash equilibria arise depending on whether agents anticipate that others may react if they deviate from the equilibrium. In the present context, a number of recent papers~\citep{voss.19,CasgrainJaimungal:20,N-V-2021} study \emph{open-loop} Nash equilibria, where each agent treats the others' trading rates as fixed stochastic processes. This means that agents do not consider at all how the others may react if they were to deviate from the equilibrium:

\begin{definition}[Open-Loop Nash Equilibrium] \label{def:Nash}
A collection of admissible trading rates $\dot\varphi=(\dot\varphi^1,\ldots,\dot\varphi^N)$ is an \emph{open-loop Nash equilibrium} if no agent $n=1,\ldots, N$ has an incentive to deviate from the equilibrium, in that
\begin{equation*}
    J^n(\dot{\varphi}^n,\dot\varphi^{-n}) \geq J^{n}(\dot\psi^n,\dot\varphi^{-n}),
\end{equation*}
for all alternative admissible trading rates $\dot{\psi}^n \in \mathcal{A}_\rho$ of agent $n$.
\end{definition}

In this paper, we extend this analysis by solving for a \emph{closed-loop} Nash equilibrium in feedback form~\citep[Chapter 5.1.2]{Carmona:16}.  To wit, each agent considers the others' feedback controls to be fixed, but takes into account how their own deviations from the equilibrium impact others' trading through their effect on the model's state variables.\footnote{The even more sophisticated notion of ``subgame perfect Nash equilibria'' is studied by~\cite*{chen.al.21} in a model where agents' holdings affect the expected returns of the risky asset. This means that even the feedback controls of the other agents are not fixed out of equilibrium, but have to solve an appropriate optimization problem, too. For the present model with instantaneous trading costs, this form of interaction does not appear to be tractable.} In the present context, single-agent optimal trading rates~\citep{GARLEANU16} and open-loop Nash equilibria~\citep{voss.19} are functions of the (exogenous) trading signal $\mu_t$ and the (endogenous) risky positions $\varphi_t=(\varphi^1_t,\ldots,\varphi^N_t)$ of the agents. We therefore naturally search for a closed-loop equilibrium in the same class of strategies. If agents unilaterally deviate from the equilibrium, this changes their positions and then in turn feeds back into the trades of the other agents.

Formally, this leads to the following notion of Nash equilibrium.

\begin{definition}[(Admissible) Feedback Controls] \label{ad-control} 
A \emph{feedback trading rate} is a function $\dot\varphi : \mathbb{R}^{1+N} \to \mathbb{R}^N$ that maps the current value of the signal $\mu_t$ and all agents' risky positions $\varphi_t=(\varphi^1_t,\ldots,\varphi^N_t)$  to the rates at which the agents' risky positions are adjusted. A trading rate is \emph{admissible} if the $\mathbb{R}^N$-valued system of controlled positions
$$
d\varphi_t = \dot\varphi(\mu_t,\varphi_t)dt
$$
has a unique solution for which $\dot\varphi=(\dot\varphi(\mu_t,\varphi_t))_{t \geq 0}$ and $\varphi$ belongs to $\mathcal{A}_\rho^N$.
\end{definition}

\begin{definition}[Closed-Loop Nash Equilibrium] \label{cl-equi} 
An admissible feedback control is a \emph{closed-loop Nash equilibrium}, if no agent $n=1,\ldots,N$ has an incentive to deviate from the equilibrium in that 
$$
J^n( \dot\varphi^n; \dot\varphi^{-n}) \geq J^n( \dot\psi^n; \dot\varphi^{-n}),
$$
for all admissible feedback trading rates $\dot{\psi}= (\dot\varphi^{1},\ldots, \dot\varphi^{n-1}, \dot{\psi}^n, \dot\varphi^{n+1}, \ldots, \dot\varphi^{N})$ where the other agents' feedback controls remain fixed.
\end{definition}


\section{Main Results}\label{sec-results} 

We are now ready to state our main result. It shows that, for sufficiently small price impact $\lambda$, there exists a symmetric closed-loop Nash equilibrium. The corresponding equilibrium trading rate is linear with respect to the signal $\mu_t$ and the agents inventories $\varphi_t=(\varphi^1_t,\ldots,\varphi^N_t)$. More specifically, it tracks an ``aim portfolio'' (a constant multiple $M_{aim}$ of the frictionless optimal holdings $\mu_t/\gamma$) with a constant (relative) trading rate $M_{rate}$. This parallels results for single-agent models~\citep{GARLEANU16}, a central planner (Section~\ref{ss:planner}), or open-loop equilibria (\cite{voss.19,CasgrainJaimungal:20} or Section~\ref{sec-open}). However the coefficients as well as the corresponding optimal value all depend on the form of the agents' strategic interaction. We discuss this in more detail in Section~\ref{ss:comparison} below. 

\begin{theorem} [Closed-Loop Nash Equilibrium]   \label{thm-nash-closed-loop}
For sufficiently small $\lambda$, there exists a symmetric closed-loop equilibrium, where 
 \begin{equation}
\label{eq-nash-control}
\dot{\varphi}^{n}_{t} =M_{rate}\left(M_{aim} \frac{\mu_t}{\gamma} -\varphi^n_t\right), \quad  n=1,\ldots,N,
\end{equation}
for some positive constants $M_{rate}, M_{aim}$. For initial positions $\varphi^{i}_{0}=0$, $i=1,...,N$ and a zero initial signal $(\mu_0=0)$, the corresponding equilibrium value of the agents' goal functionals~\eqref{def:FPGobjective} is
\begin{equation} \label{val-cl-f} 
\begin{aligned}
J^{n}(\dot\varphi^n;\dot\varphi^{-n}) =&\left(1+2\lambda N M_{rate}^{2}\frac{M_{aim}}{\gamma}\right)\frac{M_{rate} M_{aim}}{\gamma}\frac{ \sigma ^2}{\rho 
   (2 \beta +\rho ) (\beta +\rho +M_{rate} )} \\
   &\quad -\frac{1}{\rho  (2 \beta +\rho )}\left( \frac{\sigma M_{rate} M_{aim}}{\gamma}\right)^{2} \left(\lambda N 
   +\frac{\gamma+2\lambda N M_{rate}^{2}
   }{
   (\rho +2 M_{rate} ) (\beta +\rho +M_{rate} )} \right).
\end{aligned}
\end{equation}
\end{theorem}

For better readability, the lengthy proof of Theorem \ref{thm-nash-closed-loop} is deferred to Section \ref{sec-pf-cl}.  
\begin{remark} 
The coefficients $M_{rate}$ and $M_{aim}$ in Theorem \ref{thm-nash-closed-loop} are given by \eqref{eq:varconst} and \eqref{n3} in terms of the solution to the system of nonlinear equations \eqref{eq-a-system}--\eqref{eq-f-system}, \eqref{eq-a-bar-system}--\eqref{eq-c-bar-system}, which can be determined via the scalar equation~\eqref{phi-N-lambda-h}. See Section \ref{s:heuristics} for more details. By the implicit function theorem, this solution is in fact unique for sufficiently small $\lambda$, see Proposition~\ref{prop-asymptotics-ift} for more details.  
\end{remark} 
    
The constants in Theorem \ref{thm-nash-closed-loop} are determined by the root of the scalar equation \eqref{phi-N-lambda-h}, which can readily be solved numerically. However, in order to disentangle the effects of holding costs $\gamma$, trading costs $\lambda$, and competition between $N$ agents, it is also instructive to expand the solution for small $\lambda$.  The proof of this result is again deferred for better readability, see Section~\ref{sec-cl-asymp}. 

\begin{proposition}
\label{prop-asymptotic-cl}
For small price impact $\lambda \to 0$, 
\begin{equation}\label{eq-r-alpha-CL}
\begin{aligned}
M_{rate} &= \sqrt{\frac{\gamma}{\lambda}}\Delta(N)+ \mathcal{O}\left(1\right), \quad  M_{aim} = 1+ \mathcal{O}(\sqrt{\lambda}),
\end{aligned}
\end{equation}
for a nonnegative function $\Delta(N)$ that only depends on the number of agents $N$, cf.~\eqref{delta-n}. For initial positions $\varphi^{i}_{0}=0$, $i=1,...,N$ and a zero initial signal $(\mu_0=0)$, the corresponding equilibrium value of the agents' goal functional satisfies
\begin{equation}\label{eq:valuecl}
J^{n}(\dot\varphi^n; \dot\varphi^{-n}) =\frac{\sigma^2}{
 2\rho\gamma(2  \beta  +    \rho)}-\frac{\sqrt{\lambda}\sigma^2 (1 + 2 \Delta(N)^2 N) }{
  4\gamma^{3/2}\rho \Delta(N)  }  + \mathcal{O}\left(\lambda\right).
\end{equation}
\end{proposition}

Here, the function $\Delta(N)$ is determined by the unique root of the limiting version~\eqref{phi-N-lambda-limit} of the scalar equation~\eqref{phi-N-lambda-h}, which only depends on $N$ but not the other model parameters.

\section{Comparison to Other Forms of Interactions} \label{sec-comp} 

We now discuss some of the quantitative and qualitative implications of Theorem~\ref{thm-nash-closed-loop} and Proposition~\ref{prop-asymptotic-cl}. More specifically, we consider how the equilibrium trading rates and their performances depend on the number of agents and the nature of their strategic interaction.

\subsection{Central Planner}\label{ss:planner}

To better understand the closed-loop Nash equilibrium, one natural reference point is the case where the agents cooperate perfectly, in that a ``central planner'' chooses all of their controls $\dot\varphi=(\dot\varphi^1,\ldots,\dot\varphi^N)$ simultaneously in order to maximize the agents' average welfare:
\begin{equation}\label{eq-social-planner-obj}
\bar{J}(\dot\varphi) =\frac{1}{N} \sum_{n=1}^{N}J^{n}(\dot\varphi^n;\dot\varphi^{-n}).
\end{equation}
By symmetry, maximizing the average welfare of the agents is equivalent to a single agent problem with price impact parameter $N\lambda$.\footnote{Put differently, in the ``mean-field scaling'' where each agent has mass $1/N$ so that the total mass of agents does not depend on $N$, the central-planner problem is exactly the same as each single agent's.}
 Up to this adjustment the analysis is in complete analogy to~\cite{GARLEANU16}. For easy reference, Section~S1 of the supplementary appendix contains a concise derivation of these results.
  
\paragraph{Optimum}

The optimal trading rates chosen by the central planner are 
\begin{equation*} 
\dot\varphi^n_t =M^{\text{CP}}_{rate}\left(M_{aim}^{\text{CP}} \frac{\mu_t}{\gamma} - \varphi^n_t\right),  \quad n=1,...,N, 
\end{equation*}
where 
\begin{equation*}
M^{\text{CP}}_{rate}=\sqrt{\frac{\gamma}{2N\lambda} + \frac{\rho^{2}}{4}} - \frac{\rho}{2} \in (0,\infty), \quad M^{\text{CP}}_{aim}=  \frac{\sqrt{\frac{\gamma}{2N\lambda} + \frac{\rho^{2}}{4}} + \frac{\rho}{2}}{\sqrt{\frac{\gamma}{2N\lambda} + \frac{\rho^{2}}{4}}+ \frac{\rho}{2}+\beta} \in (0,1).
\end{equation*}
For zero initial positions of all agents and a zero initial signal, the corresponding optimal performance is 
\begin{equation*} 
\bar{J}(\dot\varphi)=  \frac{\sigma^{2}}{2\rho}\frac{1}{2N\lambda (\rho + 2\beta)}\left(\frac{\rho}{2}+ \beta  + \sqrt{ \frac{\gamma}{2N\lambda} + \frac{\rho^2}{4}}\right)^{-2}.
\end{equation*}

\paragraph{Asymptotics}

For small price impact $\lambda\to 0$, Taylor-expansion of these explicit formulas shows
\begin{equation}\label{eq:rateSP}
M^{\text{CP}}_{rate}=\sqrt{\frac{\gamma}{2N\lambda} } + \mathcal{O}\left(1\right) , \quad M_{aim}^{\text{CP}} = 1  +\mathcal{O}(\lambda),
\end{equation}
and
\begin{equation*}
\begin{aligned}
\bar{J}(\dot\varphi) &= \frac{1}{2\gamma}\frac{\sigma^2}{ \rho(2\beta+\rho)}-\frac{\sigma^2}{2\rho}\frac{\sqrt{2N}}{\gamma^{3/2}}\sqrt{\lambda}+\mathcal{O}\left(\lambda\right).
\end{aligned}
\end{equation*}

\subsection{Open-Loop Nash Equilibrium} \label{sec-open} 

Another important benchmark is the case of \emph{open-loop} Nash competition, studied by~\cite{voss.19,DrapeauLuoSchiedXiong:19,CasgrainJaimungal:20,N-V-2021}.

\paragraph{Equilibrium}

For an open-loop equilibrium, each agent's optimality condition can be derived in analogy to the single agent case by treating the other agents' trading rates as fixed. The resulting problem is analogous to the single-agent infinite time horizon model studied by \citep{Bouchard2018} and to the two agents finite time horizon model developed in ~\citep{voss.19}.  In the transaction cost term, the corresponding first-order condition then leads to an integral of $\lambda(2\dot\varphi^1_t+\sum_{n=2}^N \dot\varphi^n_t)$ rather than $\lambda 2\dot\varphi^1_t$ in the single-agent model. (Full details on the derivation of the open-loop equilibrium are provided in Section S2 of the supplementary appendix.) In the symmetric case where all agents are identical and their trading strategies must therefore be the same, it follows that the open-loop equilibrium is again of the same form as the single-agent model, up to replacing the constant $2$ in the single-agent results with $N+1$ rather than with $2N$ as in the central-planner model throughout. 

More specifically, in the unique open-loop Nash equilibrium each agent's optimal trading speed is given by 
\begin{equation*}
\dot\varphi^n_t = M^{\text{OL}}_{rate}\left(M^{\text{OL}}_{aim}\frac{\mu_{t}}{\gamma} - \varphi^n_t\right), \quad n=1,...,N,
 \end{equation*}
for
\begin{equation*}
 M^{\text{OL}}_{rate}= \sqrt{\frac{\gamma}{(N+1)\lambda} + \frac{\rho^{2}}{4}} - \frac{\rho}{2} \in (0,\infty), \qquad M^{\text{OL}}_{aim}= \frac{ \sqrt{\frac{\gamma}{(N+1)\lambda} + \frac{\rho^{2}}{4}} + \frac{\rho}{2}}{\sqrt{\frac{\gamma}{(N+1)\lambda} + \frac{\rho^{2}}{4}} +\frac{\rho}{2}+ \beta} \in (0,1).
\end{equation*}
For zero initial positions of all agents and a zero initial signal, the agents' corresponding common optimal value is
\begin{equation*}
\begin{aligned}
J_{\text{OL}}^{n}(\dot\varphi^n; \dot\varphi^{-n}) 
&=\left(1+2\lambda N( M^{\text{OL}}_{rate})^{2}\frac{ M^{\text{OL}}_{aim}}{\gamma}\right)\frac{ M^{\text{OL}}_{rate} M^{\text{OL}}_{aim}}{\gamma}\frac{ \sigma ^2}{\rho 
   (2 \beta +\rho ) (\beta +\rho + M^{\text{OL}}_{rate})} \\
   &\quad -\frac{1}{\rho  (2 \beta +\rho )}\left( \frac{\sigma  M^{\text{OL}}_{rate} M^{\text{OL}}_{aim}}{\gamma}\right)^{2} \left(\lambda N 
   +\frac{\gamma+2\lambda N\left( M^{\text{OL}}_{rate}\right)^{2}
   }{
   (\rho +2  M^{\text{OL}}_{rate}) (\beta +\rho + M^{\text{OL}}_{rate})} \right).
\end{aligned}
\end{equation*}
 
\paragraph{Asymptotics}

For small price impact $\lambda \to 0$, Taylor expansion of these explicit formulas shows that
\begin{equation}\label{eq:rateOL}
M^{\text{OL}}_{rate} =   \sqrt{\frac{\gamma}{(N+1)\lambda} } + \mathcal{O}\left(1\right), \quad  M^{\text{OL}}_{aim} =1 +\mathcal{O}\left(\lambda\right),
\end{equation}\
and
\begin{equation*}
\begin{aligned}
J_{\text{OL}}^{n}(\dot\varphi^n;\dot\varphi^{-n}) &= \frac{1}{2\gamma}\frac{\sigma^2}{ \rho(2\beta+\rho)} -\frac{\sigma^{2}(1+3N)}{4\rho\gamma^{3/2}\sqrt{1+N}} \sqrt{\lambda}+\mathcal{O}\left(\lambda\right).
\end{aligned}
\end{equation*}

\subsection{Comparison}\label{ss:comparison}

We now compare the equilibrium trading strategies and their performance in the closed-loop equilibrium, the open-loop equilibrium, and the central-planner solution. To focus on the impact of the agents' Nash competition rather than changes in overall risk-bearing capacity, we consider throughout the ``mean-field scaling'' where each agent has mass $1/N$, so that the total mass of agents in the economy does not change with $N$. The price impact parameter $\lambda$ is then replaced by $\lambda/N$ in the formulas above. 

\paragraph{Trading Rates}

Comparison of the asymptotic trading rates~\eqref{eq-r-alpha-CL}, \eqref{eq:rateSP}, and~\eqref{eq:rateOL} shows that, in each case, the equilibrium trading rates for small price impact track the frictionless optimal holdings $\mu_t/\gamma$, in parallel to general results for single-agent models~\citep{moreau.al.17}. 

The relative trading speed $M_{rate}$ with which this target strategy is tracked also depends in the same way on the ratio $\gamma/\lambda$ of inventory and trading costs, but scales differently with the number $N$ of agents in each model. To wit, the factor multiplying $\sqrt{\gamma/\lambda}$ in the trading rate chosen by the social planner is the constant $1/\sqrt{2}$ -- by symmetry, the social optimum that can be achieved with perfect cooperation evidently does not depend on the number of agents. The respective factor in the corresponding open-loop Nash equilibrium $\sqrt{N/(1+N)}$ is larger. This is an example of the ``tragedy of the commons'' -- without perfect cooperation, the agents overuse the common pool of liquidity available to all of them, because they only internalize the negative impact their trading has on their own execution prices but not on others'.

In the closed-loop equilibrium, the relative trading speed scales with $\Delta(N)\sqrt{N}$. As depicted in Figure~\ref{fig:1}, this multiplier lies between its counterparts for the open-loop equilibrium and the social planner solution for all values of $N$. The intuition is that, with closed-loop controls, agents at least partially slow down their trading when others are trading in the same direction. This reduces the negative externality somewhat and moves the agents closer to the social optimum. However, for all values of $N$ and irrespective of the other model parameters, these asymptotic formulas suggest that the closed-loop equilibrium is a lot closer to its open-loop counterpart than to the social planner model. At least for the symmetric setting with identical agents studied here, this provides compelling evidence for using open-loop models as more tractable approximations of the generally very involved closed-loop analysis.  

\begin{figure}[t]
\begin{center}
\includegraphics[width=0.7\textwidth]{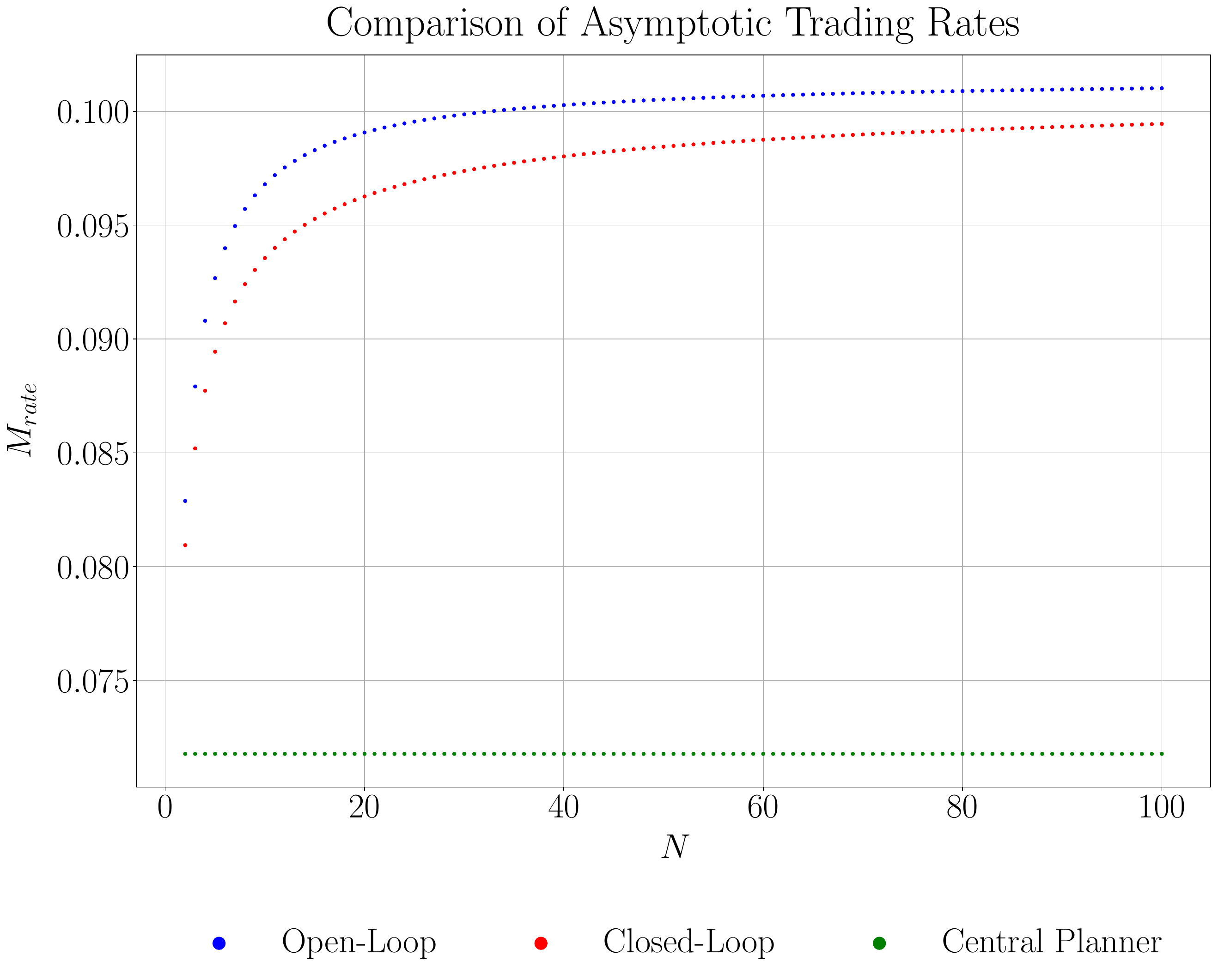}
\end{center}\caption{Leading-order asymptotics of the relative trading speeds in the central planner model (green, lowest curve), in the closed-loop equilibrium (red, middle curve) and in the open-loop equilibrium (blue, top curve) plotted against the number $N$ of agents for the parameters from Table~\ref{table:1}.}
\label{fig:1}
\end{figure}

So far, we have focused on the crisp asymptotic formulas that one obtains in the limit for small price impact $\lambda \downarrow 0$. In this regime, the differences between the central-planner solution, open-loop and closed-loop Nash equilibria only depend on the number of agents, whereas the effect of all other model parameters scales out.  We now assess the accuracy of the asymptotic approximations by comparing them to the numerical solution of the equations describing the exact closed-loop Nash equilibrium. Unlike for the asymptotic formulas, this requires realistic parameter values, both for the predictive signal and the trading and inventory cost parameters. 

To this end, we follow~\cite*{collin.al.20}. For a discrete-time version of the present model, they estimate the quadratic transaction cost parameter $\lambda$ from a proprietary dataset of real transactions executed by a large investment bank. Whereas trading costs and volatilities are constant in our model, these are modulated by a four-state Markov chain in their paper, which also acts as a trading signal by affecting expected returns. In order to translate this to our model, we average volatilities and trading costs against the stationary distribution of the Markov chain, and estimate the parameters of our Ornstein-Uhlenbeck return~\eqref{eq-ornstein-uhlenbeck} from a long simulated time series of the (centered) Markov chain.\footnote{As already mentioned above, we do not incorporate a nonzero mean-reversion level of the signal process here in order not to make the lengthy calculations for the closed-loop model even more involved.} For the inventory cost $\gamma$, we use the medium value from~\cite[Section~5.6]{collin.al.20}. These parameter values are summarized in Table~\ref{table:1}.

\begin{table}[ht]
\caption{(Daily) Model Parameters} 
\centering 
\begin{tabular}{c c} 
\hline\hline 
Parameter & Value \\ [0.5ex]
\hline 
Price volatility $\sigma_{P}$ & 0.0088\\
 Discount rate $\rho$ & 0.00004\\
Signal volatility $\sigma$ & 0.00015\\
Signal mean reversion $\beta$ & 0.070\\
Trading Cost $\lambda$ & $1.88 \times 10^{-10}$ \\ 
Inventory Cost $\gamma$ & $2.5\times 10^{-8}\times \sigma_P^2$\\ 
\hline 
\end{tabular}
\label{table:1} 
\end{table}

Figure~\ref{fig:2} compares the asymptotic approximation of $M_{rate}$ to its exact counterpart, computed numerically by solving the system of algebraic equations \eqref{eq-a-system}--\eqref{eq-f-system}. Even though the size of the portfolios under consideration here is quite large (the frictionless portfolio $\mu_t/\gamma$ has a stationary standard deviation of about two hundred million shares),  
the asymptotic approximation of the relative trading speed turns out to be almost perfect for all values of $N$.

\begin{figure}[ht]
\begin{center}
\includegraphics[width=0.7\textwidth]{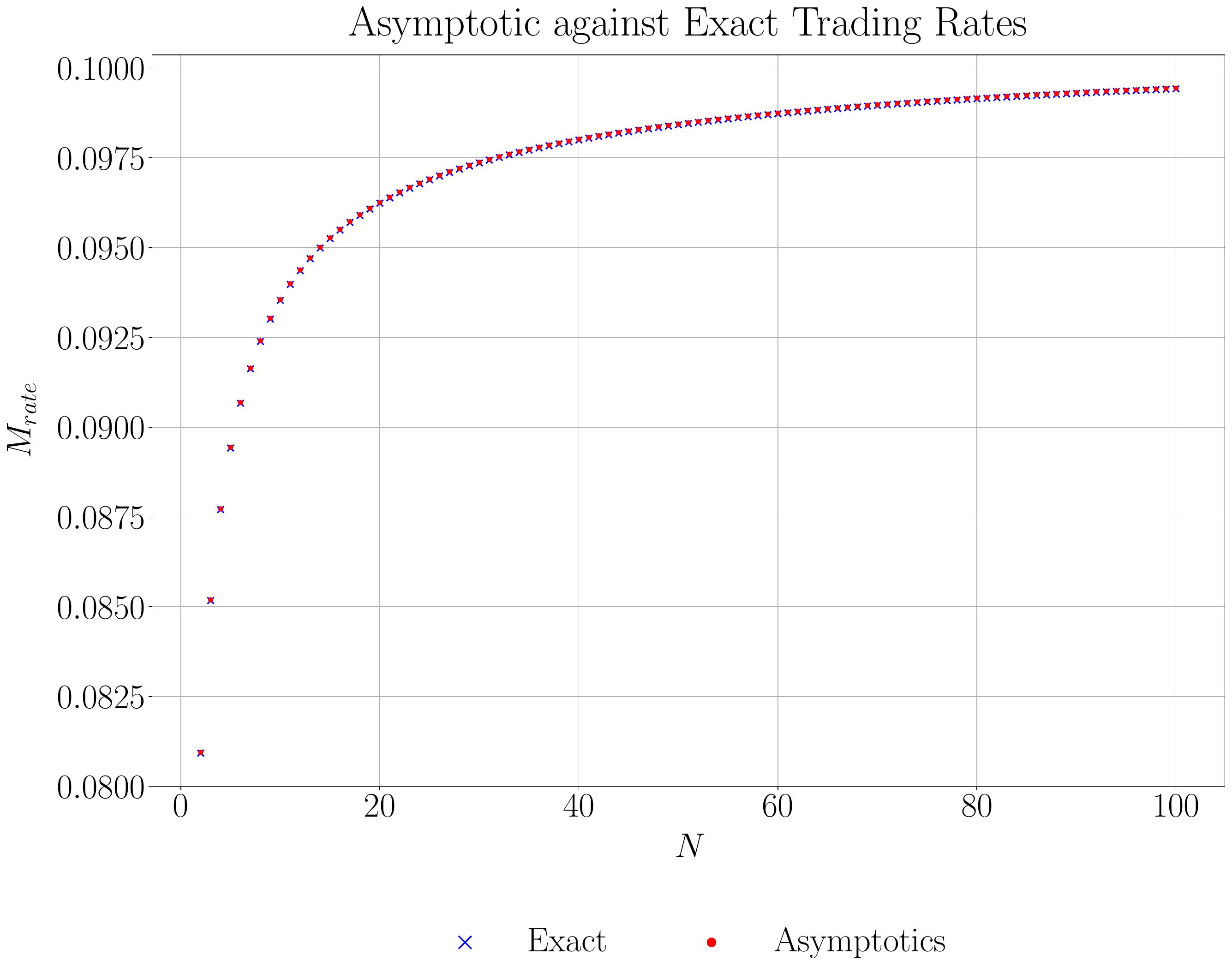}
\end{center}\caption{Relative trading speeds in the closed-loop equilibrium (blue crosses) and its asymptotic approximation (red dots), plotted against the number of agents $N$ for the parameters from Table~\ref{table:1}.}
\label{fig:2}
\end{figure}

\paragraph{Aim Portfolios}

As illustrated by Figure~\ref{fig:3}, this is not the case for the aim portfolio. Indeed, for the large portfolio sizes considered here, the multiplier $M_{aim}$ of the frictionless optimal portfolio is not close to its asymptotic value $1$, but consistently below $60\%$.\footnote{In addition to the size of the portfolios, this is due to the relatively fast mean reversion of the trading signal considered in~\cite{collin.al.20}. The dividend yield used as a predictor in \cite{barberis.00} is much more persistent, for example, so that the leading-order asymptotics are considerably more accurate in this case.} As for the relative trading speeds, the aim portfolio for the closed-loop model lies between its counterparts for the central-planner and open-loop models. The interpretation again is that trading activity is scaled back somewhat towards the social optimum, but still remains much closer to the closed-loop equilibrium for all values of $N$.

\begin{figure}[htbp]
\begin{center}
\includegraphics[width=0.62\textwidth]{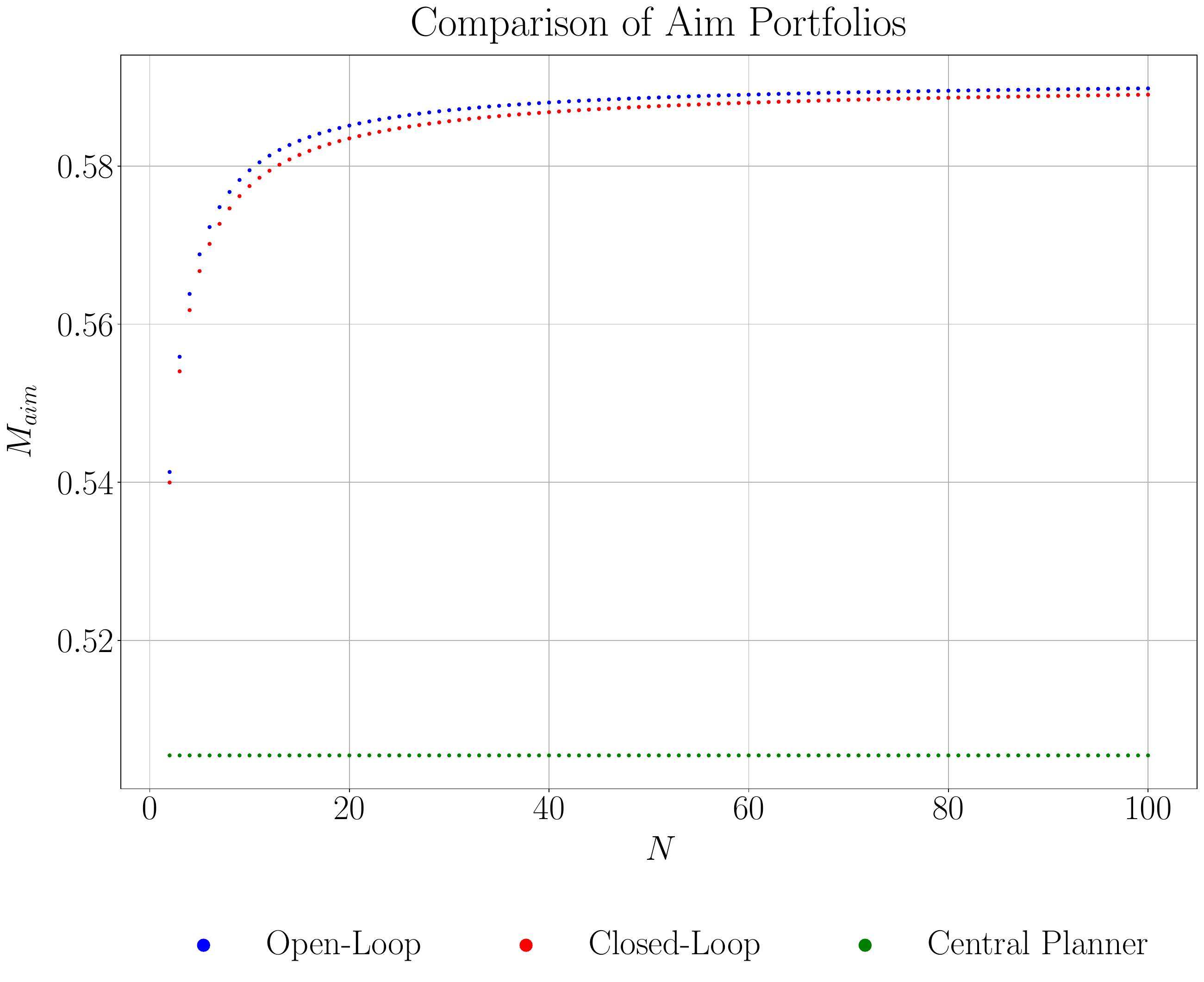}
\end{center}\caption{Multipliers $M_{aim}$ for the aim portfolio $M_{aim} \times \mu_t/\gamma$ in the closed-loop equilibrium (red, middle curve), in the central-planner model $M_{aim}^{\text{CP}}$ (green, lowest curve) and in the open-loop equilibrium, $M^{\text{OL}}_{aim}$ (blue, top curve), plotted against the number of agents $N$ for the parameters from Table~\ref{table:1}.}
\label{fig:3}
\end{figure}

\begin{figure}[htbp]
\begin{center}
\includegraphics[width=0.62\textwidth]{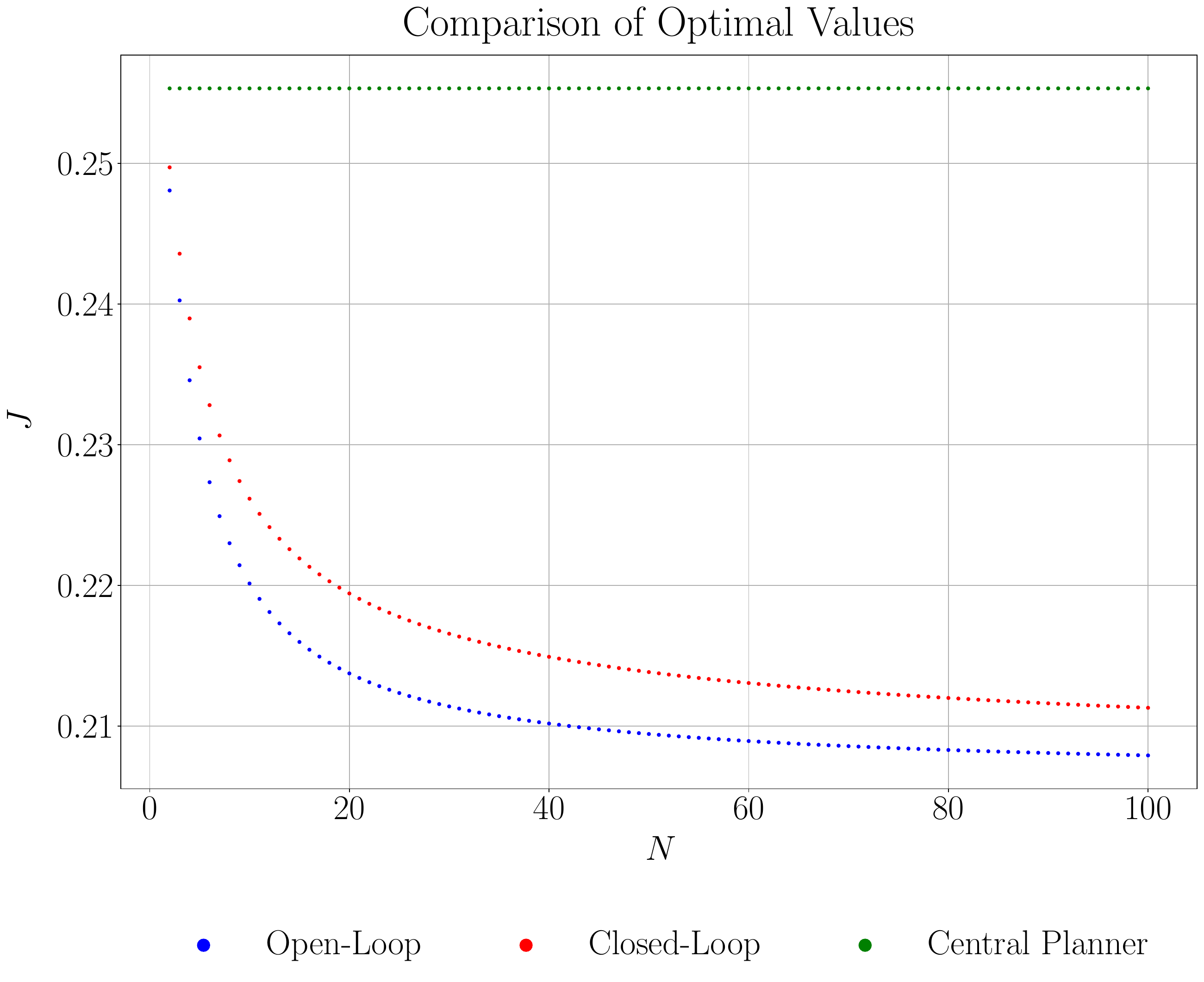}
\end{center}\caption{Optimal values in the central-planner model (green, top curve of dots), in the closed-loop equilibrium (red, middle curve), and in the open-loop equilibrium (blue, bottom curve) as fractions of the frictionless optimal value, plotted against the number of agents $N$ for the parameters from Table~\ref{table:1}.}
\label{fig:4}
\end{figure}

\paragraph{Performance}

Since the aim portfolio changes considerably relative to its frictionless counterpart, the optimal values in the different models also need to be compared using the exact formulas in the respective models rather than relying on asymptotics alone. This is illustrated in Figure~\ref{fig:4}, where we plot the optimal performances with trading costs against the number $N$ of agents in each model. To make these numbers easier to interpret, we report them as fractions of the optimal frictionless value. Recall that the closed-loop trading rates and aim portfolios move partially from the open-loop solution towards the central planner solution. Accordingly, the optimal values of the closed-loop equilibrium also fall between the open-loop model and the central-planner solution. The shortfall compared to the frictionless version of the model increases as the lack of cooperation becomes more and more severe for large $N$. This is called the  ``price of anarchy'' in the game-theory literature.

\paragraph{Mean-Field Limit?}

The numerical results reported above suggest that the open-loop and closed-loop models lead to rather similar results for realistic parameter values. An intriguing theoretical question is whether the difference vanishes completely in the ``mean-field limit'' of many small agents. 

Rigorous convergence results of this kind have recently been obtained for models with interaction through the controlled state processes~\citep{lacker.20,lacker.leflem.21} and for ``extended'' mean-field games with interaction through the agents' controls \citep{Djete21}. However in  their setup, each agent is also affected by an idiosyncratic noise, which is crucial to the proof of convergence.
 It is an important question for future research to study similar limiting results to generic ``extended'' mean-field games, where the only source of randomness is a common noise, as in the present paper. To prove this convergence result in our setting, one would have to study the large-$N$ asymptotics of the system of equations \eqref{eq-a-system}--\eqref{eq-f-system}, \eqref{eq-a-bar-system}--\eqref{eq-c-bar-system} characterizing our closed-loop equilibrium for the rescaled trading cost $\lambda/N$. We do not pursue this direction in this paper, but we report some positive numerical evidence. Specifically, in Figure~\ref{fig:5} we plot the ratio of the value functions in the closed-loop and open-loop models for numbers of agents up to $N=1000$ which conjecture the existence of \textit{some} limiting model to which the closed-loop and open-loop equilibria converge. Identifying the correct form of a limiting model that can in turn be studied directly is an important direction for future research.

\begin{figure}[htbp]
\begin{center}
\includegraphics[width=0.7\textwidth]{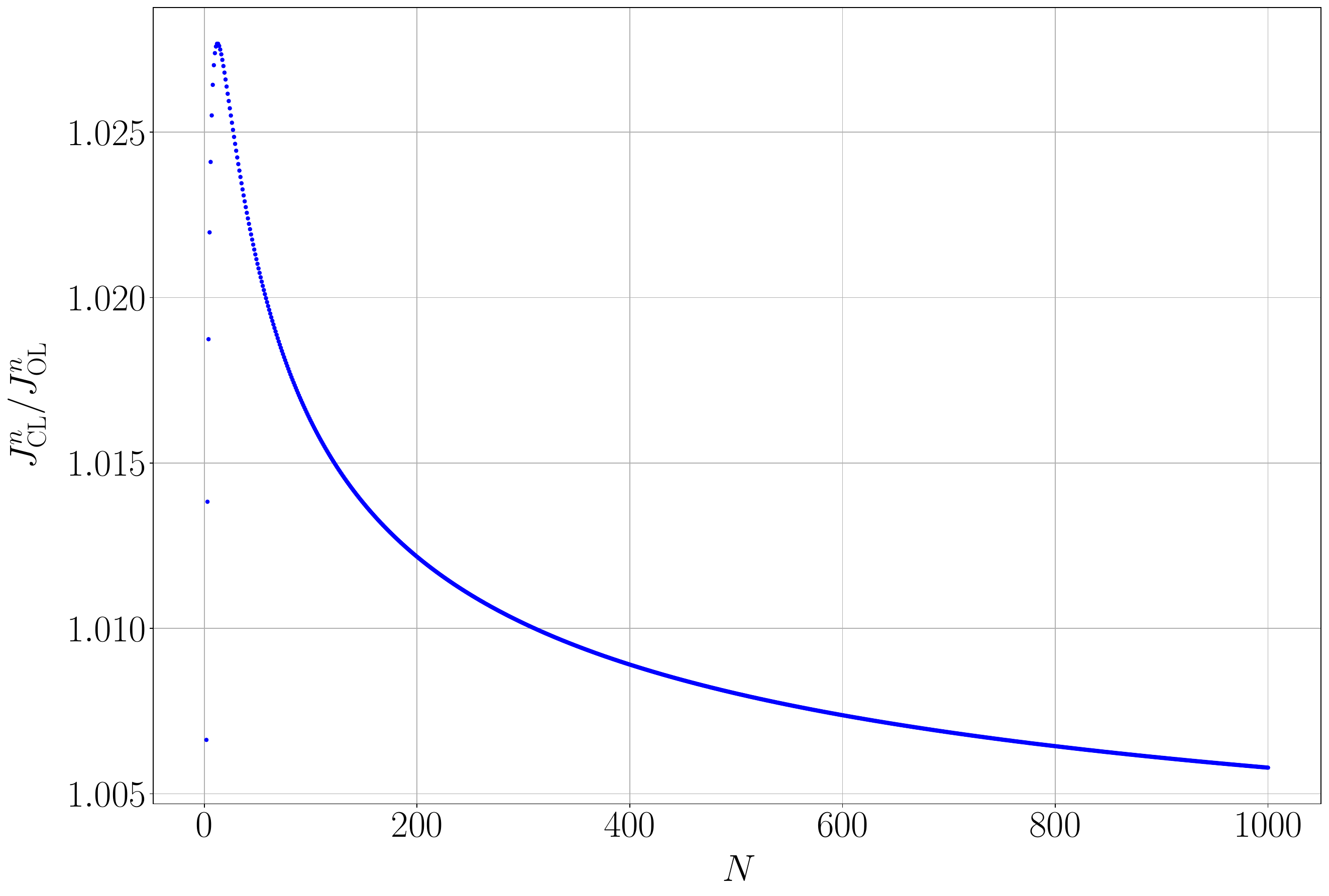}
\end{center}\caption{Ratio of the optimal values in the closed-loop and in the open-loop equilibriums, plotted against the number of agents $N$ for the parameters from Table~\ref{table:1}.}
\label{fig:5}
\end{figure}

\section{Heuristics for the Closed-Loop Equilibrium}\label{s:heuristics}

In single-agent models~\citep{GARLEANU16}, the optimal trading rate for Ornstein-Uhlenbeck returns~\eqref{eq-ornstein-uhlenbeck} is linear both in the current trading signal and in the agent's current position. In open-loop equilibria, the agents' optimal trading rates also depend on these two state variables as well as the other agents' positions in a linear fashion~\citep{voss.19,CasgrainJaimungal:20}. Accordingly, we search for closed-loop equilibria in the same linear class. 

To this end it suffices, by symmetry, to focus on the optimization problem of agent $n=1$, when the feedback trading rates of the other agents are fixed:
\begin{equation}\label{eq:control}
\dot\varphi^n_t= \bar{a} \mu_t+ \bar{b}\sum_{m \neq n} \varphi^m_t-\bar{c}\varphi^n_t, \quad n=2,\ldots,N.
\end{equation}
Here, $\bar{a},\bar{b},\bar{c}>0$, so that the above ansatz implies that agents try to reduce their own positions, but trade less to reduce price impact costs if others have the same objective. When the trading rates of agents $n=2,\ldots,N$ are fixed, agent $n=1$ faces a standard stochastic control problem of choosing their own trading rate $\dot{\varphi}^1$ to maximize 
\begin{small}
\begin{equation*}\label{eq:problem}
\mathbb{E}\left[\int_{0}^{\infty} e^{-\rho t}\left(\mu_t \varphi^1_t -\frac{\gamma}{2}(\varphi^1_t)^2- \lambda\dot\varphi^1_t\left(\dot\varphi^1_t +(N-1)(\bar{a}\mu_t +\bar{b}((N-2)\bar{\varphi}_t+\varphi^1_t)-\bar{c}\bar{\varphi}_t)\right)\right)dt\right].
\end{equation*}
\end{small}
Here, $\varphi^1_t$ and the (symmetric) positions $\bar{\varphi}_t=\varphi^2_t=\ldots=\varphi^N_t$ of the other agents have the coupled controlled dynamics
\begin{align*}
d\varphi^1_t &= \dot\varphi^1_tdt,\\
d\bar\varphi_t &= \left(\bar{a}\mu_t+\bar{b}((N-2)\bar\varphi_t+\varphi^1_t) -\bar{c}\bar{\varphi}_t\right)dt.
\end{align*}
By linearity of these dynamics and those of the exogenous state process $\mu_t$, we make the ansatz that the value function of agent $1$ is purely quadratic in the signal $m$, agent $1$'s own position $x$ and the other agents' positions $y$ (which are all the same given that they apply the same trading rates~\eqref{eq:control}):
\begin{equation}\label{eq:quadratic}
V(x,y,m)=-\frac{a}{2}x^2 + \frac{b}{2} y^2 +\frac{c}{2}m^2 -dxy +exm+fym+g,
\end{equation}
for constants $a,b,c,d,e,f,g$ to be determined.

\begin{remark}
If the Ornstein-Uhlenbeck state process~\eqref{eq-ornstein-uhlenbeck} has a nonzero mean-reversion level, then this purely quadratic ansatz needs to be extended to include three further linear terms.  The analysis below then generalizes, but involves three further equations. We therefore do not pursue this extension of the model in order not to complicate the already cumbersome calculations below.
\end{remark}

The corresponding standard infinite-horizon HJB equation is 
\be \label{eq-HJB-equation}
\begin{aligned}
\rho V =& mx -\frac{\gamma}{2}x^2 -\beta m \partial_m V +\frac12 \sigma^2 \partial^2_m V +\left[\bar{a}m+\bar{b}((N-2)y+x)-\bar{c}y \right] \partial_y V\\
& \quad +\sup_{\dot\varphi^1 \in \mathbb{R}}\big\{ -\lambda \dot\varphi^1[\dot\varphi^1+(N-1)\bar{a}m+(N-1)((N-2)\bar{b}-\bar{c})y \\
&\qquad \quad \qquad +(N-1)\bar{b}x]+\dot\varphi^1 \partial_x V\big\}.
\end{aligned}
\ee
After plugging in the quadratic ansatz~\eqref{eq:quadratic}, the pointwise maximizer can be computed as 
\begin{equation}\label{eq:FOC}
\dot\varphi^1= \frac{(e-\lambda(N-1)\bar{a})m - (a +\lambda(N-1)\bar{b})x -(d+\lambda(N-1)((N-2)\bar{b}-\bar{c}))y}{2\lambda}.
\end{equation}
After inserting this back into the HJB equation, comparison of coefficients for the terms proportional to $x^2$, $y^2$, $m^2$, $xy$, $xm$, $ym$ and constant terms in turn yields the following seven equations that pin down the coefficients $a$, $b$, $c$, $d$, $e$, $f$, $g$ of agent $1$'s value function for fixed trading rates~\eqref{eq:control} of the other agents:
\begin{align}
    0&= \frac{\rho a}{2}-\frac{\gamma}{2} -d\bar{b} + \frac{(a+ \lambda (N-1)\bar{b})^{2}}{4\lambda} \label{eq-a-system},\\
 0&=-\frac{\rho b}{2}+ b((N-2)\bar{b}-\bar{c}) + \frac{(d+\lambda (N-1)((N-2) \bar{b}-\bar{c}))^{2}}{4\lambda} \label{eq-b-system},\\
0&=\rho d + b\bar{b} - d((N-2)\bar{b}-\bar{c}) + \frac{(a+\lambda (N-1)\bar{b})(d+\lambda (N-1)((N-2) \bar{b}-\bar{c}))}{2\lambda} \label{eq-d-system},\\
0 &=-g + \frac{\sigma^{2}}{2\rho} c \label{eq-g-system},\\
0 &=-\left(\frac{\rho + 2\beta}{2}\right) c+ f\bar{a} + \frac{(e-\lambda(N-1)\bar{a})^{2}}{4\lambda}\label{eq-c-system}, \\
0 &= -(\rho + \beta) e+1 - d\bar{a} + \bar{b} f - \frac{(e-\lambda(N-1)\bar{a}) (a+\lambda(N-1)\bar{b})}{2\lambda}\label{eq-e-system}, \\
b&\bar{a} + f((N-2)\bar{b}-\bar{c}) =(\rho + \beta) f+ \frac{(e-\lambda(N-1)\bar{a}) (d + \lambda(N-1)((N-2)\bar{b}-\bar{c}))}{2\lambda}\label{eq-f-system}.
    \end{align}

A symmetric Nash equilibrium is then identified by requiring that agent $1$ has no incentive to deviate from the other agents' common controls, in that the same weights are placed on trading signals, own and others' inventories in each case. Comparison between~\eqref{eq:control} and~\eqref{eq:FOC} in turn leads to the following three additional equations:
\begin{align*}
\bar{a}&= \frac{(e-\lambda(N-1)\bar{a})}{2\lambda},\\
(N-1)\bar{b} &= -\frac{d+\lambda(N-1)((N-2)\bar{b}-\bar{c})}{2\lambda},\\
\bar{c} &= \frac{ (a +\lambda(N-1)\bar{b})}{2\lambda}.
\end{align*}
By algebraic manipulations, this system of ten nonlinear equations for ten unknowns can be reduced to a single scalar equation. To wit, first observe that
\be  \label{eq-a-bar-system}
\bar{a}=\frac{e}{(N+1)\lambda}.
\ee
Then, notice that the linear equations for $\bar{b}$ and $\bar{c}$ can be solved in terms of $a$ and $d$: 
 \begin{align}
  \bar{b}&=\frac{-aN+a+2 d}{\lambda -\lambda  N^2}, \label{eq-b-bar-system}\\  
\bar{c}&=-\frac{d-a N}{\lambda +\lambda  N} \label{eq-c-bar-system}.
\end{align}
After plugging in the expression for $\bar{b}$, \eqref{eq-a-system} can in turn be solved for $a$ in terms of $d$\footnote{Note that $a(d)$ chosen in~\eqref{eq-a-coeff} is the only root of \eqref{eq-a-system} that gives positive $\bar{a},\bar{b}$ and $\bar{c}$.}
 \begin{equation}\label{eq-a-coeff}
a(d)= \frac{d (6 N+2)-\lambda  (N+1)^2 \left(-\sqrt{-\frac{4 d^2 (3 N+1)}{\lambda ^2 (N-1) (N+1)^2}+\rho ^2+\frac{8 \gamma  N^2-4 d \rho  (3 N+1)}{\lambda 
   (N+1)^2}}+\rho \right)}{4 N^2}. 
\end{equation}
Note that $\bar{b}$ and $\bar{c}$ in~\eqref{eq-b-bar-system} and \eqref{eq-c-bar-system} depend only on $a(d)$ and $d$. Moreover, \eqref{eq-b-system} depends only on $\bar{b}$, $\bar{c}$, $d$ and $b$. Therefore, we plug $\bar{b}$ and $\bar{c}$ from \eqref{eq-b-bar-system} and \eqref{eq-c-bar-system} into  \eqref{eq-b-system} to obtain an expression for $b$ that only depends on $a(d)$ and $d$ as follows:
\begin{equation}\label{eq-b-coeff}
b(d) =\frac{2 (N-1) (-a(d) N+a(d)+2 d)^2}{(N+1) \left(4 a(d) (N-1)+2 d (N-3)+\rho  \lambda  \left(N^2-1\right)\right)}.
\end{equation}
The remaining parameters $c,e,f,g $ can in turn be sequentially expressed in terms of $d$ only as well by proceeding as follows. First, plug in the expressions for $a(d)$, $b(d)$, and the formulas for $\bar{a}$ (a linear function of $e$) and for $\bar{b}$ and $\bar{c}$ (in terms of $d$ only) into \eqref{eq-e-system} and \eqref{eq-f-system}. This leads to two linear equations for $e(d)$ and $f(d)$, that can be solved explicitly in terms of  $d$, see \eqref{eq-solution-e-f-c} for the lengthy explicit expressions. Then, after inserting these representations and the formula for $\bar{a}$ in terms of $d$,\footnote{Note that $\bar a$ is determined in terms of $d$ by \eqref{eq-a-bar-system} and our solution for $e$ as $\bar{a}(d)= \frac{e(d)}{(1+N)\lambda}.$} \eqref{eq-c-system} leads to an explicit formula for $c(d)$ in terms of $d$ (see \eqref{eq-solution-e-f-c}). Finally, the coefficient $g(d)$  can be determined in terms of $d$ by using the formula for $c(d)$ and \eqref{eq-g-system}, $g(d) = \frac{\sigma^{2}}{2\rho} c(d)$.

Finally, we use \eqref{eq-d-system} to determine an algebraic equation for the parameter $d$. Note that since $a(d)$ in \eqref{eq-a-coeff} depends only on $d$, then, $b(d)$ in \eqref{eq-b-coeff} depends only on $d$. Moreover, notice that  $\bar{b}(d)$ and $\bar{c}(d)$ as stated in \eqref{eq-b-bar-system} and \eqref{eq-c-bar-system} only depend on $d$ since they only depend on $d$ and $a(d)$. Hence, we can plug the expressions for $a(d)$, $b(d)$, $\bar{b}(d)$ and $\bar{c}(d)$ from  \eqref{eq-a-coeff}, \eqref{eq-b-coeff}, \eqref{eq-b-bar-system} and \eqref{eq-c-bar-system} into  \eqref{eq-d-system} to obtain an involved but explicit scalar equation for $d$:
 \begin{equation}
 \begin{aligned}
 \label{eq-expanded-d}
0&=\frac{-2 a(d)^2 N (N-1)^2}{\lambda  (N-1) (N+1)^2}
+\frac{a(d) (N-1) (b(d) N+b(d)+8 d N)}{\lambda  (N-1) (N+1)^2}\\& \qquad+\frac{d \left(-2 b(d) (N+1)+d (N-6) N+d+\lambda  (N-1) (N+1)^2 \rho \right)}{\lambda  (N-1) (N+1)^2}.
 \end{aligned} 
 \end{equation}
We stress that although $a(d)$ and $b(d)$ appear in \eqref{eq-expanded-d}, via the identities in \eqref{eq-a-coeff} and \eqref{eq-b-coeff},  \eqref{eq-expanded-d} is a scalar equation depending \emph{only} on the parameter $d$. In order to identify a solution $d=d(\lambda)$ of \eqref{eq-expanded-d} for fixed $N\geq 2$ and sufficiently small $\lambda$, we postulate the following factorization,  
\begin{equation}
\label{cov-d-h}
d(\lambda) = \sqrt{\gamma\lambda}\delta_N(\lambda),
\end{equation} 
for some function $\delta_N(\cdot)$ to be determined. The rescaling with $\sqrt{\lambda}$ is essential in order to obtain a nontrivial limit as $\lambda \dr 0$ in \eqref{eq-expanded-d}. The rescaling with $\sqrt{\gamma}$ leads to a limiting equation for $\lambda =0$ which does not depend on $\gamma$ but only on the number $N$ of agents. This change of variables therefore asymptotically decouples the effects of trading costs ($\lambda$), inventory costs ($\gamma$) and competition ($N$), see Lemma~\ref{prop-asymptotics-ift}(i).

With the change of variable in \eqref{cov-d-h}, the function $\delta_N(\lambda)$ is characterized as a root of a scalar equation obtained from \eqref{eq-expanded-d}:
\begin{equation}
\label{phi-N-lambda-h}
\Phi_{N}(\lambda,\delta_N(\lambda))=0.
\end{equation}
(See \eqref{eq-ift-phi}--\eqref{eq-theta} for the explicit form of $\Phi_{N}$.) In the limit $\lambda \downarrow 0$,  we show in Lemma~\ref{prop-asymptotics-ift} that $\delta_N(0)=:\delta_N^*>0$ is the unique root of this equation:
\begin{equation}\label{phi-N-lambda-limit}
\Phi_{N}(0,\delta^{*}_N)=0.
\end{equation}
In fact, any root of this equation is the root of a cubic polynomial. Cardano's method therefore leads to three explicit candidates for the roots of $\Phi_N(0,\cdot)$. Using symbolic calculations detailed in the Mathematica companion to this paper, we then verify that only one of these polynomial roots also is a root of $\Phi_N(0,\cdot)$.

Next, we also show that $\partial_{y}\Phi_{N}(\lambda,y)|_{(\lambda,y)=(0,\delta^*_N)}>0$. The implicit function theorem in turn allows us to pin down $\delta_{N}(\lambda)$ as the unique continuous function defined in a neighbourhood of $0$ such that \eqref{phi-N-lambda-h} holds. For sufficiently small $\lambda$, this yields a solution $d(\lambda)$ of \eqref{eq-expanded-d} and in turn our system of ten equations derived from the agents' optimality and consistency conditions.

The implicit function theorem also allows us to obtain the leading-order asymptotics of the agents' optimal policies and their performance. To wit, 
\begin{equation*} 
 \delta_N(\lambda) =\delta^{*}_N  + \mathcal{O}(\lambda).
 \end{equation*}
Together with \eqref{cov-d-h} , it follows that
\be \label{d-dl} 
d=\sqrt{\lambda\gamma}\delta^{*}_N+ \mathcal{O}(\lambda^{3/2}).
\ee
This allows us to obtain asymptotic expansions of all the coefficients in \eqref{eq-a-system}--\eqref{eq-a-bar-system} since these can all be expressed as functions of $d$. For example, using the expansion~\eqref{d-dl} for $d$, \eqref{eq:control}, the assumption of symmetric equilibrium and Taylor expansion lead to the expansion of the relative trading speed $M_{rate} $ from \eqref{eq-r-alpha-CL}, 
\begin{equation*} \label{constrnt-eq} 
M_{rate}   = \bar{c} - (N-1)\bar{b}  = \sqrt{\frac{\gamma}{\lambda}}\Delta(N) + \mathcal{O}(1),
\end{equation*}
where
\be \label{delta-n} 
\Delta(N)= \frac{ 1}{2 N^2} \sqrt{\frac{2 N^3-2 N^2-3 N (\delta_N^*)^2-(\delta_N^*)^2}{(N-1)  }}+\frac{(2 N+1) \delta_N^*}{2 N^2} . 
\ee
More details on the asymptotic expansions are provided in~Section \ref{sec-proof-sys-lem}.

\begin{remark} \label{algebra-comment} 
The asymptotic analysis is complicated by the fact that $\delta_N^*$ is a root of a third order polynomial (see \eqref{eq-polynomial}), which does not admit a simple expression unlike in single-agent models~\citep{GARLEANU16,moreau.al.17}. Explicit expressions for $\delta^*_N$ in terms of $N$ can be derived using Cardano's method, and these allow to verify that only one of the polynomial roots also is a root of the equation for $\delta^*_N$. However, the explicit expression (involving imaginary numbers) is too complex for deriving the analytical properties necessary to (i)  show that our candidate value function is well defined and (ii) establish a verification theorem. 

As a way out, we therefore instead use implicit properties of $\delta_N^*$ to show that the other coefficients $a$, $b$, $c$, $e$, $f$, $g$ are well defined (in terms of $d$, and therefore by \eqref{d-dl} in terms of $\delta_N^*$) and have the right signs to carry out the verification argument.

For example, in Lemma~\ref{lemma-h4-positive} we prove that for sufficiently small $\lambda$, $e(\lambda)$ is well defined and strictly positive. As follows from \eqref{eq-a-bar-system}, this is needed in order to verify that $\bar a>0$ in \eqref{constr}.  Another example appears in Lemma~\ref{lemma-alpha} where we prove that $M_{aim}  = 1+ \mathcal{O}(\sqrt{\lambda})$ for sufficiently small $\lambda $. This is essential to derive the asymptotics of the goal functionals $J^{n}(\dot\varphi^n; \dot\varphi^{-n})$ in Proposition~\ref{prop-asymptotic-cl}. As pointed out in \eqref{m-aim-exp}, the expansion of $M_{aim}$ follows from a highly nontrivial connection between $N$, $\delta^*_N$, and $\Delta(N)$ which is proved in Lemma~\ref{lemma-identities-h04-h14}, and arrises directly from properties of the polynomial root $\delta^*_N$. 
\end{remark} 

\section{Proofs} \label{sec-pf-cl} 

\subsection{Proof of Theorem~\ref{thm-nash-closed-loop}}

This section is dedicated to the proof of Theorem~\ref{thm-nash-closed-loop}. The most onerous part of the proof is to show that the optimality and consistency conditions derived in the previous section indeed have a solution with the right signs for sufficiently small $\lambda$. The lengthy proof of this result is postponed to Section~\ref{sec-proof-sys-lem} for better readability.

\begin{lemma}\label{lemma-system-value}
\begin{itemize} 
\item[\textbf{(i)}] 
For sufficiently small $\lambda>0$, there exists a solution to the system of nonlinear equation~\eqref{eq-a-system}--\eqref{eq-c-bar-system}.\footnote{This solution is in fact unique in the sense of the implicit function theorem, i.e., there is a unique limiting solution with a unique continuous extension, see Lemma~\ref{prop-asymptotics-ift} for a precise statement. }
\item[\textbf{(ii)}]  This solution $(a,b,c,d,e,f,g,\bar{a},\bar{b},\bar{c})$ to \eqref{eq-a-system}--\eqref{eq-c-bar-system} is well defined and satisfies
\be \label{constr} 
\bar{a} >0 \quad \textrm{and }  \bar{c} -(N-1)\bar b>0. 
\ee
\end{itemize} 
\end{lemma}

To prove our main result, Theorem~\ref{thm-nash-closed-loop}, we focus without loss of generality on the optimization problem of agent $1$, where the feedback controls of agents $2,\ldots,N$ are fixed according to~\eqref{eq:control}. Then, regardless of the policy chosen by agent 1, the trading rates and holdings of the other agents $n=2,\dots,N$ will be the same, so it is convenient to simplify the notation by using that
\begin{equation*}
\dot{\varphi}^2_{t} = \frac{1}{N-1} \sum_{k=2}^{N}  \dot\varphi^{k}_{t}, \quad \mbox{and}\quad  \varphi^2_{t} = \frac{1}{N-1} \sum_{k= 2}^{N}\varphi^{k}_{t}, \quad t\geq 0.
\end{equation*}
Note that $d\varphi^2_t = \dot\varphi^2_tdt$ and, by \eqref{eq:control},
\begin{equation}
\begin{aligned}
\dot\varphi^2_t &=\frac{1}{N-1}\sum_{k=2}^N(\bar{a} \mu_t - \bar{c}  \varphi^k_{t} + \bar{b}\sum_{m\neq k} \varphi^m_{t} )= \bar{a} \mu_t  +\bar{b} \varphi^1_{t}  -(\bar{c} - (N-2)\bar{b})\varphi^{2}_{t}.
\end{aligned}
\end{equation}
Using these observations, we can rewrite agent $1$'s goal functional~\eqref{def:FPGobjective} as follows:
\begin{equation}\label{eq-symmetric-cl-objective}
 \begin{aligned} 
&J^{1}(\dot \varphi ^{1}; \dot \varphi^{2})(x,y,m)\\
&= \mathbb{E}_{x,y,m}\bigg[\int_0^\infty e^{-\rho t}\bigg(\mu_t \varphi_t^{1} -\frac{\gamma}{2}\left(\varphi_t^{1}\right)^2 \\
&\qquad \qquad\quad -\lambda \dot \varphi_t^{1}\left(\dot \varphi^{1}_{t} + (N-1) (\bar{a} \mu_t +\bar{b} \varphi^{1}_{t} -(\bar{c}-(N-2)\bar{b})\varphi^{2}_{t})\right)\bigg)dt\bigg].
 \end{aligned} 
\end{equation}
Here the expectation $\mathbb{E}_{x,y,m}[\cdot]$ is taken conditional on the initial values $\varphi^{1}_{0}=x$, $\varphi^{2}_{0}=y$, and $\mu_0=m$; to ease notation, we often suppress this dependence. The corresponding value function is denoted by (here, the supremum is taken over admissible feedback controls in the sense of Definition~\ref{ad-control}):
 \be \label{val-cl} 
V^{1}(x,y,m) = \sup_{\dot \varphi^{1}} J^{1}(\dot \varphi ^{1}; \dot \varphi^{2})(x,y,m). 
 \ee

\begin{remark}
In fact, the subsequent analysis shows that deviations from the Nash equilibrium are also suboptimal among non-Markovian controls, as long as the resulting system of state equations has a sufficiently integrable solution. 
\end{remark}

We will prove Theorem \ref{thm-nash-closed-loop} using a verification argument that identifies agent 1's value function and optimal trading rate, and thereby shows that this (representative) agent has no incentive to deviate from the common feedback trading rate adopted by the other agents $n=2,\ldots,N$. As a preparation for this result, we first establish that our candidates for the equilibrium trading rates are indeed admissible. 

\begin{lemma}  \label{lem-admis} 
For the coefficients $ \bar{a},\bar{b},\bar{c}$ from Lemma~\ref{lemma-system-value}, consider the feedback trading rates
 \begin{equation} \label{opt-u-strat}  
\dot \varphi^n_t = \bar{a} \mu_t - \bar{c} \varphi^n_t +\bar{b}\sum_{m \neq n} \varphi^{m}_{t}, \quad n=1,\ldots,N, \quad t \geq 0.
  \end{equation}
 Then, the following system of linear forward SDEs
 \begin{equation}
d\varphi^{n}_{t} = \dot{\varphi}^n_t dt, \quad \varphi^{n}_{0}=x, \quad n=1,\ldots,N,
 \end{equation}
for $(\varphi^1,\ldots,\varphi^N)$ has a unique solution with $\varphi^n, \dot \varphi^n \in \mathcal A_{\rho}$ for all $n=1,\ldots,N$.
\end{lemma} 

\begin{proof} 
As all agents apply the same trading rates in the context of this lemma, the corresponding trading rates~\eqref{opt-u-strat} simplify to 
\begin{equation}\label{eq:ratessym}
\dot \varphi^n_t = \bar{a} \mu_t+ ((N-1)\bar{b}-\bar{c})\varphi^n_t, \quad n=1,\ldots,N.
\end{equation}
Put differently, the corresponding positions $\varphi^n_t$ are the unique solution of a (random) linear ODE. By the variation of constants formula, it is given by
\begin{equation}\label{eq:varconst}
\varphi^n_t =e^{-M_{rate}t}x + \int_0^t e^{-M_{rate}(t-s)}\bar{a}\mu_s ds, \quad \mbox{where } M_{rate}  = \bar{c} -(N-1)\bar b.
\end{equation}
Together with~\eqref{eq:ratessym}, it follows that
\be \label{explicit-d-phi} 
\dot\varphi^n_t= \bar{a}\mu_t -M_{rate} \int_0^t e^{-M_{rate}(t-s)}\bar{a}\mu_s ds- M_{rate}  e^{-M_{rate}t}x.
\ee
As $M_{rate}>0$ by~ \eqref{constr}, \eqref{eq:varconst}, H\"older inequality, Fubini's theorem and the second moment of the Ornstein-Uhlenbeck process~\eqref{eq-ornstein-uhlenbeck} in turn yield the required integrability of the trading rate: 
\bd\label{n4} 
\begin{aligned}   
\mathbb{E} \Big[  \int_0^{\infty} e^{-\rho t} ( \varphi^n_t)^2dt  \Big] &\leq \frac{x^{2}}{\rho + 2 M_{rate}}+ \bar{a}^2 \mathbb{E} \Big[  \int_0^{\infty} e^{-\rho t}  \left( \int_0^te^{-2M_{rate} (t-s)}ds \right) \left(  \int_0^t \mu_s^2ds \right)dt\Big] \\
&\leq\frac{x^{2}}{\rho + 2 M_{rate}} +  \frac{\bar{a}^2}{2M_{rate}} \int_0^{\infty} e^{-\rho t} \left(\int_0^t \mathbb{E}[ \mu^2_s]ds\right)dt  \\ 
&\leq \frac{x^{2}}{\rho + 2 M_{rate}}+ \frac{\bar{a}^2}{2M_{rate}} \left(\frac{\sigma^{2}}{2\beta} \right) \frac{1}{\rho},
  \end{aligned} 
\ed
and therefore $\varphi^n \in \mathcal A_{\rho}$. The corresponding integrability of the associated trading rate $\dot \varphi^n$ follows from the representation~\eqref{explicit-d-phi} along the same lines. 
  \end{proof} 
  
The exogenous signal process $\mu$ evidently has the same integrability (this can be checked, e.g., using the formula for its second moment):

\begin{lemma}\label{lemma-ornstein-uhlenbeck}
For $\rho>0$ and $\mu$ from~\eqref{eq-ornstein-uhlenbeck}, we have $\mu \in \mathcal A_{\rho}$.
\end{lemma}

The next ingredient for our verification theorem is to show that the local martingale that appears when It\^o's formula is applied to the candidate value function is in fact a true martingale.

\begin{lemma}
\label{lemma-square-integrable-martingale}
Suppose agents $n=2,\ldots,N$ use the same admissible trading rates $\dot\varphi^2=\ldots=\dot\varphi^N$ from~\eqref{opt-u-strat} and define the function $V$ as in \eqref{eq:quadratic}. Then, for any admissible trading rate $\dot\varphi_{1} \in \mathcal{A} _{\rho}$ of agent $1$, the process 
\begin{equation*}
M^{1}_{t} = \int^{t}_{0}e^{-\rho t}\partial_{m}V(\varphi_{t}^{1}, \varphi^2_t, \mu_t)dW_{t}, \quad t\geq 0, 
\end{equation*}
is a square-integrable true martingale.
\end{lemma}

\begin{proof}
From \eqref{eq:quadratic} we have $\partial_{m}V(x,y,m) = ex + fy + cm$. Accordingly, there exists a constant $C>0$ such that
\begin{equation*}
\begin{aligned}
\mathbb{E}\left[\langle M^{1}\rangle_{t}\right] &\leq C \mathbb{E}\left[\int^{t}_{0}\exp\{-2\rho t\}\Big( (\varphi_{s}^1)^2+ (\varphi^2_s)^2 + (\mu_{s})^2\Big) ds\right]<\infty,  \quad \textrm{for all } t \geq 0.
\end{aligned}
\end{equation*}
Here, we used Lemma~\ref{lemma-ornstein-uhlenbeck} and the admissibility of the trading rates in the last step. Hence, the local martingale $M^{1}$ is indeed a square-integrable martingale.
\end{proof}

Next, we show that the agents' \emph{infinite-horizon} goal functionals are indeed well defined for admissible trading rates. 

\begin{lemma}
\label{lemma-extension-dct}
For any admissible trading rates $\dot \varphi = (\dot \varphi_{1},...,\dot \varphi_{N})$, the following limit is finite: 
\be \label{dct-eq}
\begin{aligned}
&\lim_{T\to\infty} \mathbb{E}\left[\int^{T}_{0} e^{-\rho t} \left(\mu_t\varphi_t^n -\frac{\gamma}{2}(\varphi_t^n)^{2} -\lambda \dot\varphi^n_{t}\Big( \sum_{k=1}^{N} \dot \varphi_t^k \Big) \right)dt\right]  \\
& =  \mathbb{E}\left[\int^{\infty }_{0} e^{-\rho t} \left(\mu_t\varphi_t^n -\frac{\gamma}{2}(\varphi_t^n)^{2} -\lambda \dot\varphi^n_{t}\Big( \sum_{k=1}^{N} \dot \varphi_t^k \Big) \right)dt\right].
\end{aligned}
\ee
\end{lemma}

\begin{proof}
The triangle inequality and H\"older's inequality yield
\begin{align*}
F_{T}=&\left|\int^{T}_{0} e^{-\rho t} \left(\mu_t\varphi_t^n -\frac{\gamma}{2}(\varphi_t^n)^{2} -\lambda \dot\varphi^n_{t}\Big( \sum_{k=1}^{N} \dot \varphi_t^k \Big) \right)dt\right|  \\
&\leq \Big(\int^{T}_{0} e^{-\rho t} |\mu_t|^{2}dt \Big)^{1/2}\Big(\int_{0}^{T}e^{-\rho t}(\varphi_t^n)^{2}dt\Big)^{1/2} +\frac{\gamma}{2} \int^{T}_{0}e^{-\rho t} (\varphi_t^n)^{2}dt  \\
&\quad +\lambda \Big(\int_{0}^{T} e^{-\rho t} |\dot\varphi^n_{t}|^{2} dt\Big)^{1/2} \Big(\int_{0}^{T} e^{-\rho t}\Big(\sum_{k=1}^{N} |\dot \varphi_t^k|\Big)^{2} dt\Big)^{2}. 
\end{align*}
As a result, \eqref{dct-eq} follows from Lemmas \ref{lem-admis} and \ref{lemma-ornstein-uhlenbeck} and dominated convergence.  
 \end{proof}

Finally, as the last ingredient for our verification argument, we show that the candidate value function satisfies a transversality condition.

\begin{lemma}
\label{lemma-value-subexponential}
Suppose agents $n=2,\ldots,N$ use the same admissible trading rates $\dot\varphi^2=\ldots=\dot\varphi^N$ from~\eqref{opt-u-strat} and define the function $V$ as in \eqref{eq:quadratic}. Then, for any admissible trading rate $\dot \varphi^1$ of agent $1$, we have
\begin{equation}\label{eq:trans}
\liminf_{T\to\infty}\mathbb{E}\left[e^{-\rho T}V(\varphi_T^1, \varphi_T^2, \mu_T)\right] =0.
\end{equation}
\end{lemma}
    
\begin{proof}
By Lemma~\ref{lemma-ornstein-uhlenbeck} and admissibility of the trading rates, $\mu, \varphi^1$ and $\varphi^2$ all belong to $\mathcal A_\rho$. In particular,
\begin{equation}\label{eq-subexponential-sequence}
\begin{aligned} 
\liminf_{t \rr \infty}  \mathbb{E}\left[ e^{-\rho t}(\varphi_t^1)^{2}\right] &= 0, \quad \liminf_{t \rr \infty}  \mathbb{E}\left[ e^{-\rho t}(\varphi^2_t)^{2}\right] &= 0, \quad  \liminf_{t \rr \infty}  \mathbb{E}\left[ e^{-\rho t}\left(\mu_t\right)^{2}\right] &= 0. 
\end{aligned} 
 \end{equation}
The transversality condition~\eqref{eq:trans} in turn follows from the definition of $V$ in \eqref{eq:quadratic}, \eqref{eq-subexponential-sequence} and H\"older's inequality.
 \end{proof}

Now, we are ready to prove a verification theorem which shows that the function~$V$ from \eqref{eq:quadratic} indeed is the value function $V^{1}$ of agent $1$'s optimization problem~\eqref{val-cl}  for fixed feedback controls of the other agents. As a byproduct, we obtain the optimal feedback trading rate of agent $1$, which indeed coincides with the fixed feedback trading rates of the other agents $n=2,\ldots,N$ as required for a Nash equilibrium.

\begin{proposition}\label{thm-verification}
Suppose the price impact parameter $\lambda$ is sufficiently small for the coefficients $(a,b,c,d,e,f,g,\bar{a},\bar{b},\bar{c})$ to satisfy the statement of Lemma \ref{lemma-system-value}.
Define the trading rates $\dot \varphi^{2},...,\dot \varphi^{N}$ of agents $n=2,\ldots,N$ as in \eqref{opt-u-strat} and suppose agents $n=2,\ldots,N$ use the same admissible trading rates $\dot\varphi^2=\ldots=\dot\varphi^N$. Moreover, define the function $V$ as in \eqref{eq:quadratic}. Then, we have  
\be \label{v-opt}
V  \geq J^{1}(\dot\varphi^1; \dot\varphi^2), 
\ee
for all admissible feedback trading rates $\dot\varphi^1$ of agent $1$, with equality if $\dot\varphi^1$ is also given by  \eqref{opt-u-strat}. 
\end{proposition}

\begin{proof}
Let $\dot\varphi^1(x,y,m)$ be any admissible feedback control for agent $1$ and define the operator
$$
 \mathcal{L} =\dot{\varphi}^1(x,y,m)\partial_{x}+( \bar{a}\mu+\bar{b} x- (\bar{c}-(N-2)\bar{b}) y)\partial_{y}-\beta m \partial_{m} + \frac{\sigma^{2}}{2}\partial_{m}^{2}. 
$$
With this notation, It\^o's formula yields
\begin{equation}\label{eq-verification-w-martingale}
\begin{aligned}
d\Big(e^{-\rho t}V(\varphi^1_t,\varphi^2_{t},\mu_t)\Big)
&=-\rho e^{-\rho t}V(\varphi^1_t,\varphi^2_{t},\mu_t)dt + e^{-\rho t}\mathcal{L}V(\varphi^1_t,\varphi^2_{t},\mu_t)dt \\ 
&\quad +\sigma e^{-\rho t} \partial_{m} V(\varphi^1_t,\varphi^2_{t},\mu_t) dW_{t}. 
\end{aligned}
\end{equation}
The stochastic integral in \eqref{eq-verification-w-martingale} is a true martingale by Lemma~\ref{lemma-square-integrable-martingale}. Integrating and
taking expectation on both sides of \eqref{eq-verification-w-martingale} in turn leads to
\begin{equation}
\begin{split}
\label{eq-equality-verification-generator}
&V(x,y,m) -\mathbb{E}\left[e^{-\rho T}V(\varphi^1_T,\varphi^2_{T},\mu_T)\right] \\ 
&\quad=  \mathbb{E}\left[ \int^{T}_{0 } e^{-\rho t}\big(\rho V(\varphi^1_t,\varphi^2_{t},\mu_t)-\mathcal{L}V(\varphi^1_t,\varphi^2_{t},\mu_t)\Big)dt  \right].
\end{split}
\end{equation}
As the coefficients $(a,b,c,d,e,f,g,\bar{a},\bar{b},\bar{c})$ satisfy \eqref{eq-a-system}--\eqref{eq-c-bar-system} by Lemma \ref{lemma-system-value} for $\lambda$ small enough, the function $V$ solves the HJB equation \eqref{eq-HJB-equation} (compare the derivation of \eqref{eq-a-system}--\eqref{eq-c-bar-system} in Section~\ref{s:heuristics}). As a consequence, we have 
 \be \label{ineq-hjb} 
\rho  V -  \mathcal{L}V
 \geq x\mu - \lambda \dot\varphi^1( \dot\varphi^1+ (N-1)(\bar{a}m +\bar{b} x -(\bar{c}-(N-2)\bar{b})y)) - \frac{\gamma}{2}x^{2}.
  \ee
Together, \eqref{eq-equality-verification-generator} and \eqref{ineq-hjb} show
\be \label{ineq-rr} 
\begin{aligned}
&V(x,y,m) -\mathbb{E}\left[e^{-\rho T}V(\varphi^1_T, \varphi^2_{T},\mu_T)\right]  \\
&\quad\geq \mathbb{E}\bigg[\int^{T}_{0}e^{-\rho t} \big(\mu_t \varphi^1_{t}  -\frac{\gamma}{2}\big(\varphi^1_t\big)^{2} \\
&\qquad  \quad - \lambda \dot\varphi^1_t ( \dot\varphi^1_{t} + (N-1)(\bar{a} \mu_t +\bar{b} \varphi^1_{t} -(\bar{c}-(N-2)\bar{b})\varphi^2_{t}))\big)dt\bigg].
\end{aligned}
\ee
In view of Lemma~\ref{lemma-extension-dct} and the transversality condition from Lemma~\ref{lemma-value-subexponential}, we can now take the limit $T\to\infty$ on both sides of \eqref{ineq-rr} and obtain the asserted upper bound~\eqref{v-opt} for any admissible feedback control.

The admissibility of our candidate $\dot \varphi^1$ has already been established in Lemma \ref{lem-admis}. Moreover, note that when $\dot \varphi^1$ is given by \eqref{opt-u-strat} like the other agents' controls, then it follows by construction that the inequality \eqref{ineq-hjb} holds with equality (because this choice corresponds to the pointwise maximizer in the HJB equation).  Repeating the arguments leading to \eqref{ineq-rr} then shows that~\eqref{v-opt} holds with equality. 

As a result, $V$ indeed is the value function of agent $1$ given the feedback trading rates \eqref{opt-u-strat} of agents $n=2,\ldots,N$. Moreover, the same feedback trading rate is indeed optimal for agent $1$ as required for a Nash equilibrium.
\end{proof}

We finally complete the proof of the last missing items from Theorem~\ref{thm-nash-closed-loop}.

 \begin{proof}[Proof of Theorem \ref{thm-nash-closed-loop}.]
To finish the proof of Theorem~\ref{thm-nash-closed-loop}, we first show that the feedback trading rate $\dot \varphi^1$ from \eqref{explicit-d-phi} can be rewritten in the form \eqref{eq-nash-control}. This follows by defining 
\be\label{n3} 
M_{aim}  =  \frac{\bar{a}\gamma}{M_{rate}},  
\ee
and observing that \eqref{eq:varconst} and \eqref{explicit-d-phi} satisfy \eqref{eq-nash-control}. In view of the results already established in Proposition \ref{thm-verification}, it now only remains to derive the representation~\eqref{val-cl-f} for the optimal equilibrium value. Henceforth, we assume a zero initial signal $\mu_{0}=0$ and zero initial positions $\varphi^{i}_{0}= 0$, $i=1,...,N$. For the symmetric trading rates $\dot\varphi^1_t=\ldots=\dot\varphi^N_t=\dot\varphi_t$ from~\eqref{opt-u-strat}, we can (with a slight abuse of notation) rewrite the goal functionals~\eqref{def:FPGobjective} as follows: 
 \begin{equation}
J(\dot \varphi)= \mathbb{E}\left[\int_0^\infty e^{-\rho t}\Big(\mu_t \varphi_t -\frac{\gamma}{2}(\varphi_t)^2 -\lambda N (\dot \varphi_t)^{2} \Big) dt\right].
\end{equation}
The representation~\eqref{eq-nash-control} gives 
\begin{equation} 
(\dot\varphi_t)^{2} =\left(\frac{M_{rate}  M_{aim}  }{\gamma}\right)^{2}\mu_t^{2} -2M_{rate}^{2} \frac{M_{aim}  }{\gamma} \mu_t\varphi_t+ M_{rate}^{2} \left(\varphi_t\right)^{2},
\end{equation}
from which we obtain
\be \label{ghh} 
\begin{aligned}
J(\dot \varphi)&=   \mathbb{E}\Bigg[\int_0^\infty e^{-\rho t}\bigg(\left(1+2\lambda N M_{rate}^{2}\frac{M_{aim}  }{\gamma}\right)\mu_t \varphi_t\\
&\qquad \quad  -\left(\frac{\gamma}{2}+\lambda N M_{rate}^{2}\right)\left(\varphi_t\right)^2  -\lambda N \left(\frac{M_{rate} M_{aim}  }{\gamma}\right)^{2}\mu_t^{2} \bigg)dt\Bigg].
\end{aligned}
\ee
In order to derive the explicit expression for $J$ it now remains to compute the integrals 
\begin{equation*}
\begin{aligned}
\mathcal{I}_{1} &=    \mathbb{E}\left[\int_0^\infty e^{-\rho t}\mu_t\varphi_t dt\right], \quad \mathcal{I}_{2} =  \mathbb{E}\left[\int_0^\infty e^{-\rho t}\left(\varphi_t\right)^{2}dt\right],  \quad \mathcal{I}_{3} =   \mathbb{E}\left[\int_0^\infty e^{-\rho t}\mu_t^{2}dt\right].
\end{aligned}
\end{equation*}
A straightforward calculation shows that
\begin{equation*}
\begin{aligned}
\mathcal{I}_{1} &=\frac{M_{rate} M_{aim}  }{\gamma}\frac{\sigma^{2}}{2\beta}  \int^{\infty}_{0}e^{-\rho t} \int^{t}_{0} e^{-M_{rate} (t-s)} \left(e^{-\beta |t-s|} -e^{-\beta (t+s)}\right)dsdt, \\
\mathcal{I}_{2} &=\left(\frac{M_{rate} M_{aim}  }{\gamma}\right)^{2}\frac{\sigma^{2}}{2\beta}\\
&\qquad\times \int^{\infty}_{0}e^{-\rho t}\int^{t}_{0}e^{-M_{rate} (t-u)} \int^{t}_{0}e^{-M_{rate} (t-s)} \left(e^{-\beta |u-s|} - e^{-\beta (u+s)}\right)dsdudt,\\
\mathcal{I}_{3} &=\frac{\sigma^{2}}{2\beta}  \int^{\infty}_{0}e^{-\rho t}(1-e^{-2\beta t})dt.
\end{aligned}
\end{equation*}
Here, have used that
\begin{equation*} \label{cov-ou} 
\begin{aligned}
 \mathbb{E}\left[\mu_t \mu_s\right] &= \text{cov}(\mu_t\mu_s) + \mathbb{E}\left[\mu_t\right]\mathbb{E}\left[\mu_s\right]=\frac{\sigma^{2}}{2\beta}\left(e^{-\beta |t-s|} - e^{-\beta(t+s)}\right).
\end{aligned}
\end{equation*}
Evaluation of the above integrals gives 
\begin{align*}
 \mathcal{I}_{1} &=\frac{M_{rate} M_{aim}  }{\gamma}\frac{\sigma^{2}}{\rho 
   (2 \beta +\rho ) (\beta +\rho +M_{rate} )}, \\
\mathcal{I}_{2 } &=\left(\frac{M_{rate} M_{aim}  }{\gamma}\right)^{2}\frac{2 \sigma^{2}}{\rho  (2 \beta +\rho )
   (\rho +2 M_{rate} ) (\beta +\rho +M_{rate} )}, \\
   \mathcal{I}_{3} &=
\frac{\sigma^{2}
  }{2
   \beta  \rho +\rho ^2}. 
\end{align*}
Plugging these expressions for $\mathcal{I}_{1}, \mathcal{I}_{2}$ and $\mathcal{I}_{3}$ into \eqref{ghh} in turn yields the asserted representation~\eqref{val-cl-f} from Theorem~\ref{thm-nash-closed-loop}.
  \end{proof}
  
  \subsection{Proof of Lemma \ref{lemma-system-value}} \label{sec-proof-sys-lem} 

This section contains the lengthy proof of Lemma~\ref{lemma-system-value}, which states that the system of the agents' optimality and consistency conditions has a solution with the properties required for our verification result, Proposition~\ref{thm-verification}.

As derived heuristically in Section~\ref{s:heuristics}, the solution of the ten algebraic equations can be reduced to finding the root of a single complicated but explicit function $\Phi_N(\lambda,\cdot)$:
\begin{equation}\label{eq-ift-phi}
\begin{aligned}
&\Phi_{N}(\lambda,y) \\
&=\frac{1}{\lambda  (N-1) (N+1)^2}\bigg(y \left((N+1) \sqrt{\gamma  \lambda } \left(\rho  \lambda  \left(N^2-1\right)-2 \Psi_N \right)+\gamma  \lambda  ((N-6) N+1) y\right)\\
&\quad -2 \Theta_N^2 (N-1)^2
   N+\Theta_N  (N-1) \left(N \Psi_N +8 N y \sqrt{\gamma  \lambda }+\Psi_N \right)\bigg),
\end{aligned}
\end{equation}
where
\begin{equation*}
\Psi_N(\lambda,y)=\frac{2 (N-1) \left(\Theta_{N} -\Theta_{N}  N+2 y \sqrt{\gamma  \lambda }\right)^2}{(N+1) \left(\rho  \lambda  \left(N^2-1\right)+4 \Theta_{N}  (N-1)+2 (N-3) y \sqrt{\gamma 
   \lambda }\right)}
\end{equation*}
and
\begin{align}
& \Theta_N(\lambda,y) \notag\\
&=\frac{1}{4N^{2}}\Bigg[\lambda  (N+1)^2 \left(\sqrt{\rho ^2+\frac{8 \gamma  N^2-4 \rho  (3 N+1) y \sqrt{\gamma  \lambda }}{\lambda  (N+1)^2}-\frac{4 \gamma  (3 N+1) y^2}{\lambda 
   (N-1) (N+1)^2}}-\rho \right) \notag\\
   &\qquad\qquad +(6 N+2) y \sqrt{\gamma  \lambda }\Bigg]. \label{eq-theta}
\end{align}
 Note that the functions $\Theta_{N}(\lambda,y)$ and $\Psi_{N}(\lambda,y)$  correspond to the functions $a(d)$ and $b(d)$, introduced in \eqref{eq-a-coeff} and \eqref{eq-b-coeff}, after substituting the change of variable $d=\sqrt{\lambda\gamma}y$.

We now show that $\Phi_N(\lambda,\cdot)$ has a root $\delta(\lambda)$ for sufficiently small $\lambda$. Here, the limiting function $\Phi_N(0,\cdot)$ only depends on the number $N$ of agents, but not the model parameters. Its roots are also roots of a polynomial of order three, which leads to three candidates for the roots of $\Phi_N(0,\cdot)$. However, using symbolic calculations, we can show that only one of these is a root of $\Phi_N(0,\cdot)$. By verifying that the relevant derivative in this limiting point does not vanish, we can then extend existence (and, essentially also uniqueness) to sufficiently small $\lambda$ by means of the implicit function theorem.
 
 \begin{lemma}
\label{prop-asymptotics-ift}
Define $\Phi_{N}(\lambda,y)$ as in \eqref{eq-ift-phi}. Then, there exists an open set $\Lambda_{N} \subset \mathbb{R}$ containing $\lambda =0$ such that:
\begin{itemize}
\item[\textbf{(i)}]  there exists a unique point $\delta^{*}_{N}$ such that $\Phi_{N}(0, \delta^{*}_{N})=0$; it satisfies 
\begin{equation}  \label{y-str-bnd} 
\delta^{*}_{N} \in\left(0,\sqrt{2} N \sqrt{\tfrac{N-1}{3 N+1}}\right];
\end{equation}
\item [\textbf{(ii)}] there exists a unique continuously differentiable function $ \delta_{N}:\Lambda_{N}
\to\mathbb{R}$  such that 
\begin{equation*}
\Phi_{N}(\lambda, \delta_{N}(\lambda))=0, \quad \textrm{for all } \lambda\in\Lambda_{N}, \quad \mbox{where  $ \delta_{N}(0) =  \delta^{*}_{N}$.}
\end{equation*}

\end{itemize}
 \end{lemma}

\begin{proof}
(i) Our goal is to show that there is a solution of
\begin{equation}
\label{eq-phi-equal-zero}
\Phi_{N}(\lambda,y)=0
\end{equation}
for sufficiently small $\lambda \downarrow 0$. In this asymptotic regime, to first order \eqref{eq-phi-equal-zero} becomes
\begin{equation}
\label{eq-asymptotic-phi}
\gamma^{3/2}\frac{\Gamma(N,y)}{\Xi(N,y)}  + \mathcal{O}(\lambda^{1/2})= 0,
\end{equation}
where, 
\begin{equation} \label{Gamma}
\begin{aligned}
\Gamma(N,y) = &-3 N^7 \sqrt{\tfrac{2 N^3-2 N^2-3 N y^2-y^2}{(N-1) (N+1)^2}}+4 N^6 y^3+8 N^6 y^2 \sqrt{\tfrac{2 N^3-2 N^2-3 N y^2-y^2}{(N-1) (N+1)^2}}\\
&+5 N^6 \sqrt{\tfrac{2 N^3-2 N^2-3 N y^2-y^2}{(N-1) (N+1)^2}}+4 N^5 y^2 \sqrt{\tfrac{2 N^3-2 N^2-3 N y^2-y^2}{(N-1) (N+1)^2}}\\
   &+2 N^5 \sqrt{\tfrac{2 N^3-2 N^2-3 N y^2-y^2}{(N-1) (N+1)^2}}-24 N^4 y^3-20
   N^4 y^2 \sqrt{\tfrac{2 N^3-2 N^2-3 N y^2-y^2}{(N-1) (N+1)^2}}\\
   &-6 N^4 \sqrt{\tfrac{2 N^3-2 N^2-3 N y^2-y^2}{(N-1) (N+1)^2}}-8 N^3 y^3-14 N^3 y^2 \sqrt{\tfrac{2 N^3-2
   N^2-3 N y^2-y^2}{(N-1) (N+1)^2}}\\
   &+N^3 \sqrt{\tfrac{2 N^3-2 N^2-3 N y^2-y^2}{(N-1) (N+1)^2}}+18 N^2 y^3+10 N^2 y^2 \sqrt{\tfrac{2 N^3-2 N^2-3 N y^2-y^2}{(N-1)
   (N+1)^2}}\\
   &+N^2 \sqrt{\tfrac{2 N^3-2 N^2-3 N y^2-y^2}{(N-1) (N+1)^2}}+10 N y^2 \sqrt{\tfrac{2 N^3-2 N^2-3 N y^2-y^2}{(N-1) (N+1)^2}}\\
   &+2 y^2 \sqrt{\tfrac{2 N^3-2 N^2-3 N
   y^2-y^2}{(N-1) (N+1)^2}}-4 N^7 y+17 N^6 y-10 N^5 y-12 N^4 y+6 N^3 y\\
   &+3 N^2 y+12 N y^3+2 y^3,
   \end{aligned}
\end{equation}
and
\begin{equation}
\label{eq-XiNy}
\begin{aligned}
\Xi(N,y) = 4 (N-1) N^4 \Bigg(&N^2 \sqrt{\tfrac{\gamma  \left(2 N^3-2 N^2-3 N y^2-y^2\right)}{  (N-1) (N+1)^2}}\\&
-\sqrt{\tfrac{\gamma  \left(2
   N^3-2 N^2-3 N y^2-y^2\right)}{  (N-1) (N+1)^2}}+\sqrt{\gamma } N^2 y-\sqrt{\gamma } N y-\sqrt{\gamma } y\Bigg).
\end{aligned}
\end{equation}
From \eqref{eq-asymptotic-phi} it follows that we can set $\Gamma(N,y)=0$ to find the roots of \eqref{eq-phi-equal-zero} when $\lambda =0$.  
We move all the terms with the square root to one side of the equality, and the rest of them to the other side of the equality. Then, squaring both sides and reorganizing gives 
\begin{equation}
2 N^6 P_N(y)=0, 
\end{equation}
for a polynomial of order three in $y^2$, whose coefficients depend on the number $N$ of agents but not the other model parameters: 
\begin{equation}
\label{eq-polynomial}
\begin{aligned}
P_N(y^2) = &8 N^6 (y^2)^3+\left(-16 N^7+4 N^6+48 N^5-52 N^4+8 N^3+20 N^2-8 N-4\right) (y^2)^2\\ &+\left(8 N^8-20 N^7+46 N^6-112 N^5+114 N^4-4 N^3-46 N^2+8 N+6\right) y^2\\& -9 N^8+48 N^7-100 N^6+96
   N^5-30 N^4-16 N^3+12 N^2-1.
\end{aligned}
\end{equation}
The polynomial $P_N(y^2)$ has a root in the interval $(0,2 N^2 \frac{N-1}{3 N+1})$. This follows from the intermediate value theorem, since $N\geq 2$ and in turn
\begin{equation}
\begin{aligned}
\label{eq-interval-points-poly}
P_N(0) &= -(N-1)^6 (3 N+1)^2 <0, \\
 P_N\left(2 N^2\tfrac{N-1}{3 N+1}\right) &=\tfrac{(N-1)^3 (N (N (N (N (4 N (2 N-3)-13)+6)+12)+6)+1)^2}{(3 N+1)^3} >0.
\end{aligned}
\end{equation}
The equalities in \eqref{eq-interval-points-poly} can be verified by expanding the products on the right-hand sides of \eqref{eq-interval-points-poly} and comparing the corresponding result with $P_N(0)$ and $P_N\left(2 N^2\tfrac{N-1}{3 N+1}\right)$ obtained by using \eqref{eq-polynomial}.
In fact, the cubic polynomial $P_N$ has three real roots, but not all of these are roots of $\Gamma(\cdot ,N)$. Using symbolic calculations detailed in the Mathematica companion of this paper, it can be verified that only one root $\delta^{*}_{N}$ of $P_N$ satisfies $\Phi_{N}(0,\delta^{*}_{N})=0$. This completes the proof of~(i). 

(ii) Again using symbolic computations, it can be verified that $\partial_{y}\Phi_{N}(0,\delta^{*}_{N})\neq0$. (More details can be found in the Mathematica companion of this paper.) Then, by the Implicit Function theorem there exists $b_{N}>0$ and a function $\delta_{N}:[0,b_N) \to\mathbb{R}$, such that $\Phi_{N}(\lambda,\delta_{N}(\lambda))=0$ for $\lambda \in [0,b_{N})$ as desired.
\end{proof}

Now, we return to the proof of Lemma \ref{lemma-system-value}. To this end, we first introduce the following auxiliary lemma, which is proved separately in Appendix \ref{app-e-exp}.

\begin{lemma} \label{lemma-h4-positive}
For all sufficiently small $\lambda$, $e=e(\lambda)$ is well defined and strictly positive. 
\end{lemma}
  
We can now carry out the proof of Lemma \ref{lemma-system-value}:

\begin{proof}[Proof of Lemma \ref{lemma-system-value}]
(i) By reverting the steps in Section~\ref{s:heuristics} we recover the parameters~$(a,b,c,d,e,f,g,\bar{a},\bar{b},\bar{c})$ in terms of $d$. We then find a solution $d=d(\lambda, N)$ to \eqref{eq-expanded-d} for $N\geq 2$ and sufficiently small $\lambda$ as described in \eqref{cov-d-h}--\eqref{d-dl} using Lemma \ref{prop-asymptotics-ift}. Moreover, it follows that there exists a neighbourhood $\Lambda_{N}$ of $0$ such that 
\begin{equation}
\label{eq-phi-n-Y}
\Phi_{N}(\lambda,\delta_{N}(\lambda))=0,\quad  \textrm{for all } \lambda \in \Lambda_{N},
\end{equation}
which implies the existence of a solution $d=d(\lambda, N)$ to \eqref{eq-expanded-d} for sufficiently small $\lambda$.

(ii) Recall that by Lemma~\ref{prop-asymptotics-ift}(i), $\delta^{*}_{N}\in (0,\sqrt{2} N \sqrt{\frac{N-1}{3 N+1}}]$. The Implicit Function theorem as applied in the proof of Lemma~\ref{prop-asymptotics-ift} yields
\begin{equation} \label{y-exp} 
\begin{aligned}
\delta_{N}(\lambda) =\delta_{N}^{*}  + \mathcal{O}(\lambda).
\end{aligned}
\end{equation}
Using \eqref{cov-d-h} and \eqref{y-exp} we find that 
\begin{equation}\label{eq-d-expansion}
\begin{aligned}
d&=\sqrt{\lambda\gamma}\delta_{N}=\sqrt{\lambda\gamma}\delta_{N}^{*} + \mathcal{O}(\lambda^{3/2}).
\end{aligned}
\end{equation} 
Thus since $\delta_{N}^{*}>0$, \eqref{eq-d-expansion} shows that  $d>0$ for sufficiently small $\lambda$. Recall that in \eqref{eq-a-coeff} $a$ was defined in terms of $d$. Plugging \eqref{cov-d-h} in  \eqref{eq-a-coeff} for the parameter $d$, and expanding around $\lambda>0$ gives
\begin{equation} \label{a-exp} 
\begin{aligned}
&a \\
&= \frac{1 }{4 N^2} \bigg( (6 N+2) \sqrt{\gamma  \lambda } \delta_N \\ 
&\quad +\lambda  (N+1)^2 \bigg( \sqrt{\frac{8 \gamma  N^2}{\lambda  (N+1)^2}-\frac{4 (3 N+1) \rho  \sqrt{\gamma  \lambda } \delta_N}{\lambda 
   (N+1)^2}-\frac{4 \gamma  (3 N+1) \delta_N^2}{\lambda  (N-1) (N+1)^2}+\rho ^2}-\rho \bigg) \bigg)\\
&=\sqrt{\gamma \lambda } \Bigg(\frac{1}{2}\sqrt{\frac{2 N^3-2 N^2-3 N (\delta_N^{*})^2-(\delta_N^{*})^2}{(N-1) (N+1)^2}} \\
&\qquad\qquad +\frac{1}{N} \sqrt{\frac{2 N^3-2 N^2-3 N
   (\delta_N^{*})^2-(\delta_N^{*})^2}{(N-1) (N+1)^2}}\\
   &\qquad\qquad+\frac{1}{2N^{2}}\sqrt{\frac{2 N^3-2 N^2-3 N (\delta_N^{*})^2-(\delta_N^{*})^2}{(N-1) (N+1)^2}}+ \frac{3N+1}{2N^{2}}\delta_N^{*}\Bigg) + \mathcal{O}(\lambda) \\
  &=\sqrt{\gamma \lambda } \Bigg(\ol D\frac{(N+1)^{2}}{2N^{2}} +\frac{3N+1}{2N^{2}}\delta_N^{*}\Bigg) + \mathcal{O}(\lambda),
 \end{aligned}
\end{equation}
where
\be \label{d-ol} 
\ol D= \sqrt{\frac{2 N^3-2 N^2-3 N (\delta^{*}_N)^2-(\delta_N^*)^2}{(N-1) (N+1)^2}}. 
\ee
Then, for sufficiently small $\lambda$, the arguments of all the square roots above are positive so that $a$ is indeed well defined and positive. 
Next, note that for $\delta^{*}_{N}\in (0,\sqrt{2} N \sqrt{\frac{N-1}{3 N+1}}]$ we have 
\be \label{n-y-ineq} 
2 N^3-2 N^2-3 N (\delta_N^{*})^2-(\delta_N^{*})^2\geq 0. 
\ee
Using \eqref{eq-b-bar-system}, \eqref{eq-c-bar-system} ,  \eqref{eq-a-coeff}  and \eqref{cov-d-h} we get that
\begin{equation}
\label{eq-dbar-inequality}
\bar{c} - (N-1)\bar{b} = -\frac{(N-1) \left(-a N+a+2 \sqrt{\gamma  \lambda } \delta_N\right)}{\lambda -\lambda  N^2}-\frac{\sqrt{\gamma  \lambda } \delta_N-a N}{\lambda
   +\lambda  N}. 
\end{equation}
Together with $a$ given in terms of $\delta_{N}$ by \eqref{a-exp},  we obtain a power series expansion for the left-hand side of \eqref{eq-dbar-inequality} around $\lambda=0$:
\begin{equation} \label{constrnt-eq} 
\bar{c} - (N-1)\bar{b}  = \sqrt{\frac{\gamma}{\lambda}}\Delta(N) + \mathcal{O}(1),
\end{equation}
where $\Delta (N)>0$ was defined in \eqref{delta-n}. Therefore, for sufficiently small $\lambda$ it holds that
\be \label{constr-c} 
\bar{c} - (N-1)\bar{b}>0.
\ee  
It remains to verify \eqref{constr}, it now remains to prove that $ \bar a>0$. In view of \eqref{eq-a-bar-system}, we have
 $$
\bar{a} = \frac{e}{(1+N)\lambda}  .
$$
so it suffices to establish $e>0$. From Lemma  \ref{lemma-h4-positive} it follows that for sufficiently small $\lambda$, the constant $e$ is strictly positive and well defined. Therefore, the constant $\bar a$ is also strictly positive and well defined.

It remains to prove that $(b,c,f,g)$ are well defined. Since $a,d>0$ for sufficiently small $\lambda$, $b$ in \eqref{eq-b-coeff} is well defined. From the expression for $f$ in \eqref{eq-solution-e-f-c} it follows that in order to verify that $f$ is well defined we need to show that all the denominators appearing in \eqref{eq-solution-e-f-c} are different from zero when $\lambda$ is small, namely 
\begin{align} 
h_3(\lambda)-\beta-\rho \neq  & 0, \label{h3-eq}  \\
h_4(\lambda)  \lambda  (N+1)^2 \left(4 a (N-1)+2
   d (N-3)+\lambda  \left(N^2-1\right) \rho \right) \neq & 0. \label{h4-eq}
\end{align} 
In view of Lemma \ref{lemma-h4-positive} and since $e= 1/h_{4}$ by \eqref{eq-solution-e-f-c}, we have $h_4(\lambda)>0$, $a>0$, $d>0$ for sufficiently small $\lambda$ and in turn \eqref{h4-eq}. (The case $N=2$ needs to be handled separately -- in this case, it follows from \eqref{eq-a-coeff} that for $\lambda$ small enough we have $a\geq d/2$, which together with $h_4(\lambda)>0$ gives \eqref{h4-eq}.) 

Using the expression for $h_{3}$ in \eqref{eq-h-constants} together with \eqref{eq-d-expansion} and \eqref{a-exp} we find that for sufficiently small $\lambda$ we have 
\begin{equation*}
\begin{aligned} 
h_{3} &= \sqrt{\frac{\gamma}{\lambda}}\left(\frac{\left(-N^2+N+1\right) \delta_N^{*}-\left(N^2-1\right) \sqrt{\frac{  \left(2 N^3-2 N^2-3 N (\delta_N^{*})^2-(\delta_N^{*})^2\right)}{(N-1) (N+1)^2}}}{
   (N-1) N^2}\right) + \mathcal{O}(1)<0,
\end{aligned} 
\end{equation*}
 where we also used \eqref{n-y-ineq} in the last step. This proves \eqref{h3-eq}, so that $f$ is well defined.  
 Next note that $g$ is well defined if $c$ is well defined, hence in order to complete the proof we need to prove that $c$ is well defined. From \eqref{eq-c-system} we have 
\be \label{c-interms} 
 c = \left(\frac{\rho + 2\beta}{2}\right)^{-1} \Big( f\bar{a} + \frac{(e-\lambda(N-1)\bar{a})^{2}}{4\lambda} \Big),
 \ee
 since we already showed that $(e,f,\bar{a})$ are well defined, the result follows.
 \end{proof}
 
\begin{remark}
\label{rmk-uniqueness}
Note that there exists a unique solution to the system \eqref{eq-a-system}--\eqref{eq-c-bar-system} for which \eqref{constr} and  \eqref{y-str-bnd} are satisfied. Specifically, this unique solution can be characterised as the unique root of the polynomial \eqref{eq-polynomial}  satisfying \eqref{y-str-bnd} while being also the root of  \eqref{eq-asymptotic-phi}. This solution identifies a unique closed-loop Nash equilibrium in the class of linear strategies.  The polynomial \eqref{eq-polynomial} may have other roots, than the one corresponding to the closed-loop Nash equilibrium, however, they either do not satisfy \eqref{y-str-bnd} or are not a root of  \eqref{eq-asymptotic-phi}.  We recall that we have already initially excluded all the solutions corresponding to the other root of \eqref{eq-a-coeff} as the corresponding solution did not lead to the correct sign in at least one of the parameters among $\bar{a},\bar{b}$ and $\bar{c}$ (see the footnote before \eqref{eq-a-coeff}). 
\end{remark}

\appendix
\section{Technical Proofs}

This appendix collects a number of elementary but cumbersome proofs used in the body of the paper.

\subsection{ Proof of Proposition \ref{prop-asymptotic-cl} and Lemma \ref{lemma-h4-positive}} \label{sec-cl-asymp} 
\label{app-e-exp}


Before proving  Lemma \ref{lemma-h4-positive}, we introduce two intermediate lemmas which will help us to disentangle the dependence of the parameter $e$ on $\lambda$.

First, we derive the small-$\lambda$ asymptotics of $h_4$ from \eqref{eq-h-constants}.

 \begin{lemma}
\label{lemma-h4-expansion}
For all $\lambda$ sufficiently small we have,
\begin{equation} \label{h-4-exp} 
h_{4} = \sqrt{\frac{\gamma}{\lambda}}h^{0}_{4} + \mathcal{O}(1),
\end{equation}
where
\begin{equation}
\label{eq-h04}
\begin{aligned}
h^{0}_{4} &= \frac{\ol{D}(\delta_{N}^{*})+\delta_{N}^{*}}{N}\\ &\quad -\frac{(N+1) (\delta_{N}^{*}-\ol{D}(\delta_{N}^{*}) (N-1))^2 }{8 N^2 \left(\bar{D
}(\delta_{N}^{*})   (N-1)^2+\left(N^2-N-1\right) \delta_{N}^{*}\right)^2}\\
&\quad \times \left(\ol{D}(\delta_{N}^{*}) (N-1) \left(N^2+2 N-3\right)+\left(3 N^2-4 N-3\right) \delta_{N}^{*}\right).
\end{aligned}
\end{equation}
\end{lemma}

\begin{proof}
From the expansions of $a$ and $d$ in \eqref{a-exp} and \eqref{eq-d-expansion} we get,
\begin{equation}
\label{eq-expansion-first-term}
\begin{aligned}
h_{4,1}=\frac{(N-1)d+2a N}{\lambda  (N+1)^2} &= \left(\frac{\ol{D}(\delta_{N}^{*})+ \delta_{N}^{*}}{N}\right)\sqrt{\frac{\gamma}{\lambda}} +\mathcal{O}(1).
\end{aligned}
\end{equation}
Define \begin{equation}
h_{4,2} = \label{eq-second-term-h4}
\frac{2 (2 d-(N-1)a)^2 \left(-a \left(N^2+2 N-3\right)+4 d-\lambda  \left(N^2-1\right) \rho \right)}{h_2 \lambda  (N+1)^2 \left(2h_2 -\lambda (N^{2}-1)(2\beta + \rho)\right)}. 
\end{equation}
We expand $h_{4,2}$ by using the power series of $a$ and $d$ and the definition of $h_{2}$ in \eqref{eq-h-constants} to obtain
 \begin{equation}
 \label{eq-expansion-upsilon}
 \begin{aligned}
 h_{4,2} &= -\sqrt{\frac{\gamma}{\lambda}}
\frac{(N+1) (\delta_N^{*}-\ol{D}(\delta_N^{*}) (N-1))^2 }{8 N^2 \left(\ol{D}(\delta_N^{*}) (N-1)
   (N+1)+\left(N^2-N-1\right)\delta_N^{*}\right)^2}\\ 
  &\qquad \times   \left(\ol{D}(\delta_N^{*}) (N+1) \left(N^2+2 N-3\right)+\left(3 N^2-4 N-3\right)\delta_N^{*}\right)
  +\mathcal{O}\left(1\right).
\end{aligned}
 \end{equation}
Plugging \eqref{eq-expansion-first-term} and \eqref{eq-expansion-upsilon} into~\eqref{eq-h-constants} then indeed yields 
 \begin{equation}
 h_{4} = h_{4,1}+h_{4,1} +\beta+\rho=  \sqrt{\frac{\gamma}{\lambda}}h_{4}^{0} + \mathcal{O}(1), 
 \end{equation}
 with $h_{4}^{0}$ defined in \eqref{eq-h04}.
\end{proof}
From Lemma \ref{lemma-h4-expansion} we observe that in order to prove Lemma \ref{lemma-h4-positive}, we need to derive the sign of  $h_4^0$, which in turn will give us the sign of $e=1/h_4$. We define the following function, 
\begin{equation}
\label{eq-chi-N}
\begin{aligned}
&\chi_{N}(y) \\
&= \frac{\ol{D}(y)+y}{N}
\\& \quad -\frac{(N+1) \left(\left(N^3+3 N^2-N-3\right) \ol{D}(y)+\left(3 N^2-4 N-3\right) y\right) \left(y-(N-1) \ol{D}(y)\right)^2}{8 N^2 \left(\left(N^2-1\right) \ol{D}(y)+\left(N^2-N-1\right)
   y\right)^2}\\&
   \quad -\frac{2 N^2}{(N+1) \left((N+1) \ol{D}(y)+(2 N+1) y\right)}.
\end{aligned}
\end{equation}
In the following lemma we argue that  
\begin{equation*}
h_{4}^{0} -\frac{1}{(N+1)\Delta(N)} = \chi_{N}(\delta_{N}^{*}), 
\end{equation*}
and we also show that $\chi_{N}(\cdot)$ vanishes once we plug-in $ \delta_{N}^{*}$. This is a key ingredient in the proof of the sign of $h_4$ as suggested by \eqref{h-4-exp}.

\begin{lemma}
\label{lemma-identities-h04-h14}
For $h_{4}^{0}$ as in \eqref{eq-h04} and $\Delta$ given by \eqref{delta-n}, we have
\begin{equation*}
 h_{4}^{0} = \frac{1}{(N+1)\Delta(N)}.
 \end{equation*}
\end{lemma}
 \begin{proof}
By substituting the expressions for $h_{4}^{0}$ from \eqref{eq-h04} and $\Delta$ from \eqref{delta-n}, we obtain
\begin{equation*}
h_{4}^{0} -\frac{1}{(N+1)\Delta(N)} = \chi_{N}(\delta_{N}^{*}). 
\end{equation*}
We will show that for any $N\geq 2$, $\delta^{*}_{N}$  satisfies
\begin{equation}
\label{eq-chi-roots-YN}
\chi_{N}(\delta^{*}_{N}) =0, 
\end{equation}
which will prove the result. 
 
First, we aggregate all the terms appearing on the right hand side of \eqref{eq-chi-N} in one fraction, that is
\begin{equation*}
\chi_{N}(\delta^{*}_{N}) = \frac{\mathcal{S}_{1}(\delta^{*}_{N})}{\mathcal{S}_{2}(\delta^{*}_{N})},
\end{equation*} 
where, 
\begin{equation*}
\begin{aligned}
 \mathcal{S}_{1}(y) &= \mathcal{R}_{1} + \mathcal{R}_{2},\\
 \mathcal{R}_{1}(y) &= -4 N^9 y^2+4 N^8 y^4+47 N^8 y^2-6 N^7 y^4-52 N^7 y^2-66 N^6 y^4-90 N^6 y^2+34 N^5 y^4\\ &+108 N^5 y^2+144 N^4 y^4+39 N^4 y^2+12 N^3 y^4-36 N^3
   y^2-82 N^2 y^4-12 N^2 y^2-9 N^9\\&+33 N^8-42 N^7+18 N^6+3 N^5-3 N^4-42 N y^4-6 y^4,\\
   \mathcal{R}_{2}(y) &= 14 (N+1) N^7 y^3 \ol{D}-22 (N+1) N^6 y^3 \ol{D}-58 (N+1) N^5 y^3 \ol{D}+56 (N+1) N^4 y^3 \ol{D}\\&+68 (N+1) N^3 y^3 \ol{D}-22 (N+1) N^2
   y^3 \ol{D}-9 (N+1) N^8 y \ol{D}+44 (N+1) N^7 y \ol{D}\\&-54 (N+1) N^6 y \ol{D}-6 (N+1) N^5 y \ol{D}+37 (N+1) N^4 y \ol{D}-6 (N+1)
   N^3 y \ol{D},
   \\&-6 (N+1) N^2 y \ol{D}-30 (N+1) N y^3 \ol{D}-6 (N+1) y^3 \ol{D}, \\
    \mathcal{S}_{2}(y)&=2 (N-1) N^2 (N+1) (\ol{D} (N+1)+2 N y+y) \left(\ol{D} \left(N^2-1\right)+\left(N^2-N-1\right) y\right)^2.
 \end{aligned}
\end{equation*}
and we abbreviate $\ol{D}(y)$ to $\ol{D}$. Hence, in order to prove \eqref{eq-chi-roots-YN} we need to show that
\begin{equation*}
 \frac{\mathcal{S}_{1}(\delta^{*}_{N})}{\mathcal{S}_{2}(\delta^{*}_{N})} = 0, 
 \end{equation*}
which is implied by $\mathcal{S}_{1}(\delta^{*}_{N})= 0$ and $\mathcal{S}_{2}(\delta^{*}_{N})\neq 0$.  This is proved by a direct substitution and by Lemma~\ref{prop-asymptotics-ift}(i) which identifies $\delta^{*}_{N}$.  
\end{proof}

Now we are ready to prove Lemma \ref{lemma-h4-positive}. 

 \begin{proof}[Proof of Lemma \ref{lemma-h4-positive}] 
From Lemmas \ref{lemma-h4-expansion} and \ref{lemma-identities-h04-h14} it follows that for all sufficiently small $\lambda$,
\begin{equation} \label{h4-exp} 
h_{4} = \sqrt{\frac{\gamma}{\lambda}}h_{4}^{0} + \mathcal{O}(1) = \sqrt{\frac{\gamma}{\lambda}}\frac{1}{(N+1)\Delta(N)}+\mathcal{O}(1). 
\end{equation}
From \eqref{delta-n} we have that $\Delta(N)>0$. As argued in the paragraph after \eqref{eq-b-coeff}, an expression for $e(d)$ can be obtained by lengthy algebraic computations: $e=\frac{1}{h_{4}}$ (see \eqref{eq-solution-e-f-c}). This completes the proof.
\end{proof}

The following intermediate result is an important step in proving Proposition \ref{prop-asymptotic-cl}.

\begin{lemma} \label{lemma-alpha} 
For $\lambda $ sufficiently small we have
\begin{equation}
M_{rate} = \sqrt{\frac{\gamma}{\lambda}}\Delta(N)+ \mathcal{O}\left(1\right), \quad M_{aim}  = 1+ \mathcal{O}(\sqrt{\lambda}).
\end{equation}
for the nonnegative function $\Delta(N)$ defined in \eqref{delta-n}.
\end{lemma}
\begin{proof} 
The representation of $M_{rate} $ in \eqref{n3} and \eqref{constrnt-eq} yield
\begin{equation} \label{r-n-eq}
M_{rate}   = \sqrt{\frac{\gamma}{\lambda}}\Delta(N) + \mathcal{O}(1), 
\end{equation}
where $\Delta(N)$ is given by \eqref{delta-n} and does not depend on $\lambda$.
From \eqref{n3}, \eqref{eq-a-bar-system} and \eqref{eq-solution-e-f-c} we have 
\begin{equation}
\label{eq-a-bar-coefficient-h4}
\begin{aligned}
M_{aim} &= \gamma \frac{\bar{a}}{M_{rate} }= \gamma \frac{e}{ \lambda (N+1)M_{rate} }=  \frac{\gamma}{ \lambda (N+1)h_{4} M_{rate} }. 
 \end{aligned}
\end{equation} 
Using \eqref{r-n-eq} and \eqref{h4-exp} we obtain that, for $\lambda$ small enough,
\begin{equation} \label{m-aim-exp} 
\begin{aligned}
M_{aim}  &= \frac{\gamma}{\lambda(N+1)}\frac{1}{\left(\sqrt{\frac{\gamma}{\lambda}}h^{0}_{4}  + \mathcal{O}(1)\right)}\frac{1}{\left(\sqrt{\frac{\gamma}{\lambda}}\Delta(N) + \mathcal{O}(1)\right)}= \frac{1}{(N+1)h_{4}^{0}\Delta(N)} + \mathcal{O}(\sqrt{\lambda}).
\end{aligned}
\end{equation}
Finally, plugging in the result of Lemma \ref{lemma-identities-h04-h14} for $h_{4}^{0}$, we get  
\begin{equation}
M_{aim} = 1+ \mathcal{O}(\sqrt{\lambda})
\end{equation}
as asserted.
\end{proof}
\begin{proof}[Proof of Proposition \ref{prop-asymptotic-cl}]
The asymptotic expansions for $M_{rate} $ and $M_{aim}$ were already proved in Lemma \ref{lemma-alpha}.

In view of this result, it now remains to derive the asymptotics of the value function. 
For $J^{n}$ as in \eqref{val-cl-f} we define 
\be \label{j-cl-dc} 
J^{n}(\dot{\varphi}^{n},\dot{\varphi}^{-n}) = \bar w_{1}(\lambda) - \bar w_{2}(\lambda) -\bar w_{3}(\lambda), 
\ee
where 
\begin{equation*}
\begin{aligned}
\bar w_{1}(\lambda) &= \left(1+2\lambda N(M_{rate} )^{2}\frac{M_{aim}  }{\gamma}\right)\left(\frac{M_{rate} M_{aim}  }{\gamma}\right)\frac{\sigma^{2}}{\rho 
   (2 \beta +\rho ) (\beta +\rho +M_{rate} )},\\
\bar w_{2}(\lambda) &=  \lambda N \left(\frac{M_{rate} M_{aim}  }{\gamma}\right)^{2}\frac{\sigma^{2}
  }{2
   \beta  \rho +\rho ^2}, \\
\bar w_{3}(\lambda) &=\left(\frac{\gamma}{2}+\lambda N\left(M_{rate} \right)^{2}\right)\left(\frac{M_{rate} M_{aim}  }{\gamma}\right)^{2}\frac{2 \sigma^{2}}{\rho  (2 \beta +\rho )
   (\rho +2 M_{rate} ) (\beta +\rho +M_{rate} )}.
\end{aligned}
\end{equation*}
Using \eqref{eq-r-alpha-CL} we can write 
\begin{equation*}
\begin{aligned}
M_{rate}  &= \sqrt{\frac{\gamma}{\lambda}}\Delta(N)+ r_{1}(N)+ \mathcal{O}\left(\sqrt{\lambda} \right), \quad  M_{aim}  = 1+ a_{1}(N)\sqrt{\lambda}+  \mathcal{O}\left(\lambda \right),
\end{aligned}
\end{equation*}
for some explicit functions $r_{1}, a_{1}$ of $N$. Here, the explicit dependence on $N$ is omitted in what follows to ease notation.  Using Taylor expansion, we can in turn derive power-series expansions of $\bar w_{1}$, $\bar w_{2}$ and $\bar w_{3}$ around $\lambda=0$,
\begin{equation*}  
\begin{aligned}
\bar w_{1}(\lambda) &=\frac{(1 + 2 \Delta^2 N) \sigma^2}{\gamma \rho (2 \beta + \rho)} +\frac{\sigma^2}{\Delta \gamma^2 \rho (2\beta + \rho)} \big( a_1 \Delta \gamma + 4 a_1 \Delta^3 \gamma N + 
     4 \Delta^2 \sqrt{\gamma}  N r_1 \\
     &\quad  - 
      \sqrt{\gamma} \beta - 
     2 \Delta^2 \sqrt{\gamma}  N\beta - 
      \sqrt{\gamma} \rho - 
     2 \Delta^2 \sqrt{\gamma}  N \rho \big) \sqrt{
   \lambda}+O(\lambda),  \\
\bar w_{2}(\lambda) &= \frac{\Delta^2 N \sigma^2}{\gamma \rho (2 \beta+ \rho)} + \frac{
 2 N \sigma^2(a_1 \Delta^2 \gamma + 
    \Delta \sqrt{\gamma}r_1) }{
 \gamma^2 \rho (2 \beta+ \rho)} \sqrt{\lambda}+O(\lambda),\\
\bar w_{3}(\lambda) &= \frac{(1 + 2 \Delta^2 N) \sigma^2}{2 \gamma \rho (2 \beta + \rho)} +\frac{\sigma^2}{4 \Delta \gamma^{3/2} \rho (2 \beta + \rho)}\big(4 a_1 \Delta \sqrt{\gamma}  (1 + 2 \Delta^2 N) - 
   2 \beta \\
   &\quad  + 2 \Delta^2 N (4 r_1 - 2 \beta - 3 \rho) - 
   3 \rho\big) \sqrt{\lambda} +O(\lambda).
\end{aligned}
\end{equation*}
After inserting these expansions into \eqref{j-cl-dc}, we obtain the asserted leading-order asymptotics~\eqref{eq:valuecl} of the optimal value. More details on these calculations can be found in the Mathematica companion of this paper. 
\end{proof}

\subsection{Identities for the Proofs of Lemmas~\ref{lemma-system-value} and  \ref{lemma-h4-positive}}
\label{app-identities}
In this section we provide some identities for the coefficients of the system \eqref{eq-a-system}--\eqref{eq-f-system}, which were used in Section \ref{s:heuristics} and in the proof of Lemma~\ref{lemma-system-value}. 

The coefficients $e$, $f$ and $c$ are given by,
\begin{equation}
\label{eq-solution-e-f-c}
\begin{aligned}
e&= \frac{1}{h_{4}},\\
f &=\frac{1}{\lambda(N+1) h_4(h_3-(\rho+\beta)) } \\
     & \quad\times \left( d+h_3 \lambda  (N-1) -\frac{2 h_1^2 (N-1)}{     (N+1) \left(4 a (N-1)+2
   d (N-3)+\lambda  \left(N^2-1\right) \rho \right)} \right),\\
c&=\frac{2 \left(-\frac{\frac{\left(\frac{1}{h_4}-\frac{N-1}{h_4 (N+1)}\right) \left(d+h_3 \lambda  (N-1)\right)}{2 \lambda }-\frac{2 h_1^2 (N-1)}{h_4 \lambda  (N+1)^2
   \left(4 a (N-1)+2 d (N-3)+\lambda  \left(N^2-1\right) \rho \right)}}{h_4 \lambda  (N+1) \left(-\beta +h_3-\rho \right)}-\frac{\left(\frac{1}{h_4}-\frac{N-1}{h_4
   (N+1)}\right)^2}{4 \lambda }\right)}{-2 \beta -\rho },
\end{aligned}
\end{equation}
where
\begin{equation}
\label{eq-h-constants}
\begin{aligned}
h_{1} &=  2 d-(N-1)a, \\
h_{2} &= 2 a (N-1)+d (N-3)+\lambda  \left(N^2-1\right) (\beta +\rho), \\ 
h_{3}&=-\frac{ 2a(N-1)+(N-3)d}{\lambda(N-1)(N+1)}, \\ 
h_{4} &= \frac{2 (2 d-(N-1)a)^2 \left(-a \left(N^2+2 N-3\right)+4 d-\lambda  \left(N^2-1\right) \rho \right)}{h_2 \lambda  (N+1)^2 \left(2h_2 -\lambda (N^{2}-1)(2\beta + \rho)\right)}+\frac{(N-1)d+2a N}{\lambda  (N+1)^2} \\
   &\quad   +\beta+\rho.
\end{aligned}
\end{equation}

 \bigskip
\paragraph{Data Availability Statement.}
Data sharing is not applicable to this article as no new data were created or analyzed in this study.


\end{document}